\documentclass{article} 
\usepackage{iclr2027_conference,times}

\usepackage[utf8]{inputenc} 
\usepackage[T1]{fontenc}    
\usepackage{hyperref}       
\usepackage{url}            
\usepackage{booktabs}       
\usepackage{amsfonts}       
\usepackage{nicefrac}       
\usepackage{microtype}      
\usepackage{xcolor}         

\usepackage{wrapfig,lipsum,booktabs}
\usepackage{multirow}
\usepackage[american]{babel}
\usepackage{mathtools}
\usepackage{tikz} 
\usetikzlibrary{arrows.meta}
\usepackage{xspace}
\usepackage[font = small]{subcaption}
\usepackage{amsmath,amsthm,amsfonts,amssymb}
\usepackage{comment}
\usepackage[linesnumbered,ruled,noend]{algorithm2e}
\usepackage{parskip}
\usepackage{eso-pic}

\newtheorem{problem}{Problem}
\newtheorem{proposition}{Proposition}
\newtheorem{theorem}{Theorem}

\newtheorem{definition}{Definition}

\newtheorem{corollary}{Corollary}
\newtheorem{lemma}{Lemma}

\usepackage{amsmath,amsfonts,bm}

\def\eqref#1{equation~\ref{#1}}

\def\1{\bm{1}}

\def\rr{{\textnormal{r}}}

\DeclareMathAlphabet{\mathsfit}{\encodingdefault}{\sfdefault}{m}{sl}
\SetMathAlphabet{\mathsfit}{bold}{\encodingdefault}{\sfdefault}{bx}{n}

\newcommand{\R}{\mathbb{R}}

\newcommand{\reals}{\mathbb{R}}
\newcommand{\naturals}{\mathbb{N}}
\newcommand{\s}{\mathcal{S}}
\newcommand{\A}{\mathcal{A}}
\newcommand{\G}{\mathbb{G}}

\newcommand{\mdp}{M = (\s, \A, p, p_0, r, H)}
\newcommand{\mdpb}{\bar{M} = (\bar{\s}, \bar{\A}, \bar{p},\bar{p}_0, \bar{r}, H)}

\title{Sample Complexity of Equivariant Reinforcement Learning}

\author{%
  Rayan Mazouz\thanks{Corresponding author \texttt{rayan@newtheory.ai}}~, 
  ~Haibo Zhao,
  ~Chris Hillar,
  ~Christian Shewmake
  \\
  New Theory AI\\
  San Francisco, CA 94110 
}

\iclrfinalcopy 
\begin{document}

\maketitle

\begin{abstract}
Reinforcement learning (RL) is a powerful framework for robotic control, yet its practical application is often hindered by high sample complexity. This is particularly restrictive in physical domains where interaction data is costly. While the world often exhibits geometric and physical symmetries, standard RL algorithms typically fail to exploit this structure. In this paper, we demonstrate that exploiting group symmetries significantly reduces the sample complexity of RL. Focusing on finite-horizon Markov decision processes, we find that leveraging homomorphisms induced by group symmetries significantly reduces the theoretical upper and lower bounds on the number of environment interactions required to reach an optimal return. We further extend these bounds to continuous state and action spaces, providing corresponding sample-complexity guarantees under appropriate regularity assumptions. Beyond theory, we validate our findings through controlled experiments and demonstrate the advantages of symmetry-aware policy learning on high-dimensional continuous robotic simulations. Our results show that integrating symmetry into the learning pipeline yields substantial gains in sample efficiency and performance, offering a principled path toward more data-efficient robotics.
\end{abstract}

\section{Introduction}

Reinforcement learning (RL) is a prominent framework for learning decision-making policies directly from interaction, with notable successes in robotics~\cite{levine2016end, kalashnikov2018qtopt, mazouz2022safety}. Recent advances in learning-enabled systems have shifted the paradigm from carefully engineered control and perception pipelines toward data-driven approaches that learn from high-dimensional sensory inputs. While these methods have demonstrated strong empirical performance, their deployment in real-world robotic systems remains constrained by poor sample efficiency and limited generalization~\cite{zhang2025improving, ngo2025scaling}. One source of this inefficiency is that many environments possess geometric structure that is not explicitly used by standard learning architectures \cite{van2020mdp}. When symmetry-related configurations are represented as distinct inputs, an agent may need to relearn the same behavior across redundant representations of the environment, increasing the amount of data required for learning. This motivates the use of equivariance as a principled mechanism for exploiting latent structure in RL.

In parallel, a growing body of work in geometric deep learning has shown that encoding symmetry through equivariant architectures can improve generalization and sample efficiency~\cite{bronstein2017geometric, bronstein2021geometric, cohen2016group, monti2017geometric, papillon2025beyond, shewmake2023group, vadgama2025probing, islam2025platonic}. These advances have primarily focused on supervised and representation learning, where the role of symmetry can be studied without the additional complications of exploration, sequential dependence, and policy-induced distribution shift~\cite{zhang2016understanding, kumar2020conservative}. In RL, symmetry-aware methods and Markov Decision Process (MDP) homomorphisms have also been explored~\cite{van2020mdp, li2006towards, ravindran2003smdp}, but existing work does not quantify how symmetry-induced reductions translate into sample-complexity gains in discrete and continuous domains. In this work, we address this gap by studying the role of symmetry as an inductive bias. Our central premise is that symmetries of the MDP induce equivalence classes over states and actions. When these equivalences are exploited, the agent can share information across equivalent problems rather than learning over redundant indexed representations, reducing the effective size of the learning problem. 

Fig.~\ref{fig:headline} illustrates this reduction through an orbit-collapse view. A homomorphism $h$ maps the original MDP $M$ to a reduced MDP $\bar{M}$ by collapsing symmetry-induced equivalence classes into representative states and actions.
In the figure, each colored orbit in $M$ contains several rotated versions of the same underlying ring configuration, while the corresponding state in $\bar{M}$ represents the entire orbit. The optimal policy found in the reduced MDP corresponds to an optimal policy in the original MDP~\cite{van2020mdp}. Thus, the agent may learn over the smaller quotient problem while still obtaining a policy that is optimal for the original space. This reduction is directly reflected in the sample complexity: the number of episodes required to reach the same performance level decreases from $\mathcal{\tilde{O}}(k_M)$ in the original MDP to $\mathcal{\tilde{O}}(k_{\bar{M}})$ in the quotient MDP.

\usetikzlibrary{calc}

\definecolor{ringblue}{RGB}{78,121,167}
\definecolor{ringgreen}{RGB}{89,161,79}
\definecolor{ringred}{RGB}{225,87,89}
\definecolor{ringneutral}{RGB}{210,210,210}
\definecolor{ringlight}{RGB}{180, 180, 180} 
\newcommand{\RingStateA}[3]{%
  \begin{scope}[shift={(#1,#2)}, rotate=#3]

    \def\rr{0.28}

    \foreach \k in {1,...,8}{
      \coordinate (p\k) at ({90-45*(\k-1)}:\rr);
    }

    \draw[ringlight,line width=0.9pt]
      (p1)--(p2)--(p3)--(p4)--
      (p5)--(p6)--(p7)--(p8)--cycle;

    \node[
      circle,
      fill=ringred,
      draw=black,
      line width=0.3pt,
      minimum size=5pt,
      inner sep=0pt
    ] at (p1) {};

    \node[
      circle,
      fill=ringgreen,
      draw=black,
      line width=0.3pt,
      minimum size=5pt,
      inner sep=0pt
    ] at (p5) {};

    \foreach \k in {2,3,4,6,7,8}{
      \node[
        circle,
        fill=ringneutral,
        draw=black,
        line width=0.3pt,
        minimum size=5pt,
        inner sep=0pt
      ] at (p\k) {};
    }

  \end{scope}
}

\newcommand{\RingStateB}[3]{%
  \begin{scope}[shift={(#1,#2)}, rotate=#3]

    \def\rr{0.28}

    \foreach \k in {1,...,8}{
      \coordinate (q\k) at ({90-45*(\k-1)}:\rr);
    }

    \draw[ringlight,line width=0.9pt]
      (q1)--(q2)--(q3)--(q4)--
      (q5)--(q6)--(q7)--(q8)--cycle;

    \foreach \k in {1,3}{
      \node[
        circle,
        fill=ringred,
        draw=black,
        line width=0.3pt,
        minimum size=5pt,
        inner sep=0pt
      ] at (q\k) {};
    }
    
    \node[
      circle,
      fill=ringgreen,
      draw=black,
      line width=0.3pt,
      minimum size=5pt,
      inner sep=0pt
    ] at (q5) {};
    
    \foreach \k in {2,4,6,7,8}{
      \node[
        circle,
        fill=ringneutral,
        draw=black,
        line width=0.3pt,
        minimum size=5pt,
        inner sep=0pt
      ] at (q\k) {};
    }
    
  \end{scope}
}

\begin{figure}[t]
\centering
\resizebox{0.7\textwidth}{!}{%
\begin{tikzpicture}[x=1.15cm,y=1.15cm]

\tikzset{
  orbitA/.style={draw=ringblue, line width=2pt},
  orbitB/.style={draw=ringred, line width=2pt},
  mapline/.style={draw=black!65, line width=0.75pt},
  repA/.style={circle, fill=ringblue, draw=ringblue, minimum size=6.5pt, inner sep=0pt},
  repB/.style={circle, fill=ringred, draw=ringred, minimum size=6.5pt, inner sep=0pt}
}

\draw[line width=0.85pt, rounded corners=22pt]
  (-8.35,-1.85)
  -- (-8.35, 2.55)
  -- (-0.75, 2.55)
  -- (-0.75,-1.85)
  -- cycle;

\draw[line width=0.85pt, rounded corners=20pt]
  (1.25,-1.65)
  -- (1.25, 2.55)
  -- (3.45, 2.55)
  -- (3.45,-1.65)
  -- cycle;

\node[font=\normalsize] at (-5.70,3.45) {MDP $M$};
\node[font=\normalsize] at ( 2.35,3.45) {Quotient MDP $\bar{M}$};

\begin{scope}[shift={(-5.85, 1.20)}, rotate=17]
  \draw[orbitA] (0,0) ellipse [x radius=1.55, y radius=1.0];

  \coordinate (T1) at (-2.02, 0.13);
  \coordinate (T2) at ( 0.87, 0.67);
  \coordinate (T3) at ( 1.57,-0.48);

    \RingStateA{-0.75}{-1.2}{0}
    \RingStateA{ -0.87}{ 1.15}{45}
    \RingStateA{ 1.77}{-0.48}{110}
\end{scope}

\begin{scope}[shift={(-3.55,-0.50)}, rotate=17]
  \draw[orbitB] (0,0) ellipse [x radius=1.55, y radius=1.0];

  \coordinate (B1) at (-1.93,-0.25);
  \coordinate (B2) at ( 0.18, 0.74);
  \coordinate (B3) at ( 1.86,-0.31);

     \RingStateB{-1.93}{-0.15}{20}
    \RingStateB{ 0.18}{ 0.64}{70}
    \RingStateB{ 1.86}{-0.31}{135}
\end{scope}

\node[repA] (RepTop) at (2.25, 1.45) {};
\node[repB] (RepBot) at (2.25,-0.45) {};

\RingStateA{2.80}{1.45}{0}
\RingStateB{2.80}{-0.45}{0}

\draw[-{Latex[length=2.2mm]}, line width=0.8pt]
  (-0.30,0.70) -- (0.95,0.70);

\node[font=\normalsize] at (0.33,0.97) {$h$};

\node[font=\small] at (0.33,0.35) {homomorphism};

\draw[mapline] ($(T1)+(+1.5,-0.60)$) -- (RepTop);
\draw[mapline] ($(T2)+(0.4,+0.07)$) -- (RepTop);

\draw[mapline] ($(B1)+(1.5,-0.25)$) -- (RepBot);
\draw[mapline] ($(B2)+(1.2, 0.08)$) -- (RepBot);

\node[font=\normalsize] at (-5.65,-2.95)
  {$\text{Sample Complexity } \tilde{O}(k_M)$};

\node[font=\normalsize] at ( 2.30,-2.95)
  {$\text{Sample Complexity } \tilde{O}(k_{\bar{M}})$};

\end{tikzpicture}
}
\caption{
The homomorphism $h$ maps equivalent states in $M$ to representative states in the quotient MDP $\bar{M}$, reducing the effective sample complexity from $\mathcal{\tilde{O}}(k_M)$ to $\mathcal{\tilde{O}}(k_{\bar{M}})$ with $k_{\bar{M}} \ll k_M$.}
\label{fig:headline}
\end{figure}
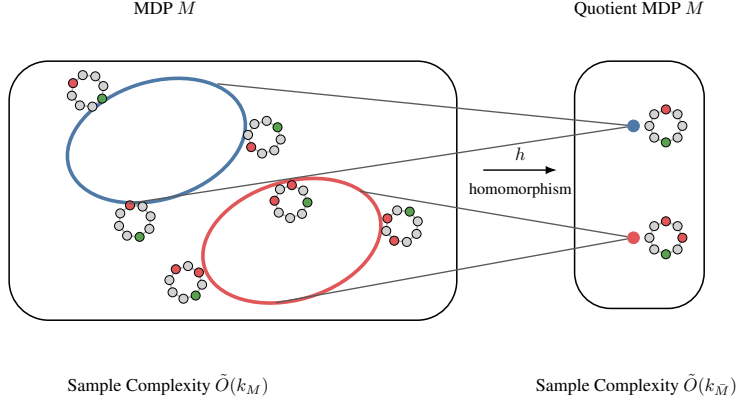

In this work, to formally reason about sample complexity, we formalize MDPs that admit group symmetries via a \emph{graph-theoretic} framework. This perspective exposes the symmetry-induced equivalence classes and provides the basis for constructing the corresponding quotient MDP (Fig.~\ref{fig:headline}). We show that learning on the quotient MDP requires fewer episodes to obtain a near-optimal policy for the original MDP than learning directly over the unreduced state-action space. Building on the PAC framework for episodic RL~\cite{dann2015sample}, we derive \emph{upper} and \emph{lower} bounds that make the dependence of sample complexity on symmetry explicit. We further extend this analysis to \emph{continuous state-action} spaces, deriving sample-complexity guarantees in terms of symmetry-reduced covering numbers under suitable regularity assumptions. These bounds show that, under suitable structural conditions, symmetry changes the effective scaling of the learning problem, yielding improvements that grow with the number of redundant representations identified by symmetry. Beyond theory, we demonstrate the practical implications through designed graph-structured experiments and broader RL benchmarks. By comparing standard and equivariant learning algorithms under identical settings, we show that exploiting latent structure improves learning efficiency and performance.

\paragraph{Contributions.}
In summary, the main contributions of this paper are as follows:
\begin{itemize}
    \item A formal framework for analyzing the sample complexity of reinforcement learning under symmetry through the lens of graph-based MDPs.
    \item Sample-complexity guarantees for continuous state-action spaces in terms of symmetry-reduced covering numbers.
    \item Upper and lower sample-complexity bounds that make the role of symmetry explicit, showing that exploiting symmetry changes the effective scaling of the learning problem.
    \item Empirical results and benchmarks demonstrating that these theoretical predictions translate into consistent improvements in learning efficiency and performance.
\end{itemize}

\paragraph{Related Work.}
\emph{Sample complexity} in RL has been extensively studied in multi-armed bandits, infinite-horizon MDPs, and episodic fixed-horizon settings~\cite{lai1985asymptotically, auer2002finite, kakade2003sample, strehl2006pac}. Early PAC analyses established finite-sample guarantees for exploration in MDPs, often with strong dependencies on the horizon, state space, and action space, with subsequent work improving these dependencies through refined optimistic and upper confidence bound methods~\cite{auer2009near, jaksch2010near, fiechter1994efficient}. Sample complexity has also been studied in continuous state-action spaces through regularity assumptions and covering arguments~\cite{maran2024projection}, while recent work has considered scaling laws relating data and compute~\cite{rybkin2025value}. However, these analyses do not quantify how group symmetries reduce the effective complexity of the learning problem. A complementary line of work studies structural abstractions, including MDP homomorphisms~\cite{li2006towards, ravindran2003smdp} and symmetry-aware representations such as homomorphic networks and group-equivariant RL~\cite{van2020mdp, finzi2021practical}. Similarly, graph-based representations have been widely used to model structured decision-making problems~\cite{battaglia2018relational,jiang2018graph,albrecht2024multi}. While these approaches exploit symmetry and structure algorithmically, they do not characterize the resulting reduction in sample complexity. This paper closes this gap by deriving symmetry-dependent sample-complexity guarantees for discrete and continuous state-action spaces.

\section{Problem Formulation}

In this work, we are interested in the analysis of the sample complexity of RL algorithms when the underlying environment contains symmetry-induced structure. We begin by reviewing the standard finite-horizon Markov Decision Process (MDP) setting, in accordance with the following definition. \\


\begin{definition}[Markov Decision Process]
\label{def:mdp}
A Markov Decision Process (MDP) is a tuple $\mdp$, where $\s$ and $\A$ are finite state and action spaces. Here, $\A$ denotes the ambient action set; in settings with state-dependent admissibility, the transition and reward functions are evaluated only on admissible pairs $(s,a)$ with $a\in\A(s)$. The function $p : \s \times \A \times \s \to [0,1]$ is a stationary transition kernel satisfying $\sum_{s' \in \s} p(s' \mid s,a) = 1$ for all admissible pairs $(s,a)$. The initial state distribution is $p_0 : \s \to [0,1]$ with $\sum_{s \in \s} p_0(s) = 1$. The reward function $r$ is possibly time-dependent and takes values in $\reals$. The horizon $H \in \naturals_{\geq 1}$ denotes the finite decision horizon.
\end{definition}

The agent interacts with this MDP over episodes of fixed horizon $H$. At time $t \in \{1,\dots,H\}$, the agent observes a state $s_t \in \s$ and selects an action $a_t = \pi_t(s_t)$ according to a (possibly time-dependent) policy $\pi = \{\pi_t\}_{t=1}^H$, where $\pi_t : \s \to \A$. The next state is sampled as $s_{t+1} \sim p(\cdot \mid s_t, a_t)$, and the initial state is drawn as $s_1 \sim p_0$.

At each time $t$, the agent receives a reward drawn from a distribution with mean $r_t(s_t, a_t)$, determined by the reward function. It is assumed that the reward function $r$ is known to the agent, whereas the transition kernel $p$ is unknown. The performance of a policy $\pi$ on $M$ is defined as the \emph{total expected reward} of an episode
$
    R_{M}^{\pi} = \mathbb{E}
    \!
    \left[
    \sum_{t=1}^H r_t(s_t,a_t)
    \right].
$
For a policy $\pi$ and transition kernel $p$, the sequence $\{s_t\}_{t=1}^{H}$ is a Markov process with a well-defined probability measure \cite{bertsekas1996stochastic}.
The expectation in $R_M^\pi$ is taken with respect to this measure.
For this setup, sample complexity can be studied through \emph{probably approximately correct} (PAC) bounds. \\

\begin{definition}[Sample Complexity]
\label{def:sc}
Let $M$ be an MDP as in Definition~\ref{def:mdp}, and let $\pi^\star$ denote an optimal policy with reward value $R_M^{\pi^\star}$. For accuracy $\epsilon > 0$ and confidence $\delta \in (0,1)$, the \emph{sample complexity} of a learning algorithm is the number of episodes in which the policy executed by the algorithm is $\epsilon$-suboptimal, i.e.,
$R_M^{\pi_t} < R_M^{\pi^\star} - \epsilon$,
with probability at least $1-\delta$.
\end{definition}

The goal of this work is to understand how this sample complexity changes when the MDP contains symmetry-induced redundancy in its representation. In what follows, we formalize the problem.\\

\begin{problem}
    \label{problem}
    Let $M$ be a finite-horizon MDP. For given $\epsilon > 0$ and $\delta \in (0,1)$, we seek to quantify how exploiting symmetry-induced redundancy changes the sample complexity.
\end{problem}

\paragraph{Approach}
Our approach to Problem~\ref{problem} is to make explicit the transition structure underlying the MDP. 
This perspective views the MDP as an action-labeled transition graph: states form the nodes, and actions label the possible transitions between states. We then study settings in which this transition graph admits group symmetry. These symmetries identify redundant state-action representations, reducing the effective size of the learning problem. This leads to upper and lower sample-complexity bounds that depend explicitly on the available group symmetry, and provides the principled basis for comparing structured learning with learning over the unreduced representation.

\section{Preliminaries}
\label{sec:preliminaries}

In this section, we show that for symmetric environments, standard MDP formulations suffer from an explosion of the state space by naively treating structurally equivalent configurations as distinct entities. To formalize this redundancy, we use graph representations of the environment, which naturally expose its connectivity and symmetries. We then relate these symmetries to MDP homomorphisms.

\subsection{Graph-Based MDP Formulation}

While a discrete MDP can be represented by its transition graph, we formulate the MDP directly over a directed graph $\G=(V,E)$ representing the environment's connectivity. In this setting, an MDP state is a configuration of features across the vertices of $\G$, together with the agent's spatial location. \\

\begin{definition}[Graph-Valued State Space]
\label{def:graph_state_space}
Let $X$ be a finite alphabet of node features, represented as one-hot vectors. The graph-valued state space is
$ \mathcal{S}:=X^V\times V.$ Each state \(s\in\mathcal{S}\) is written as $ s=(x,q), $ where \(x:V\to X\) assigns a value \(x(v)\in X\) to each node \(v\in V\), and \(q\in V\) denotes the current position of the agent.
\end{definition}

Intuitively, picture $\G$ as a game board (e.g., a Sokoban room or a chessboard): $x$ describes the current arrangement of the pieces (the state of the board), $q$ marks where the agent is standing. The topology $\G$ is fixed across all states; the labelling $x$ and position $q$ are what change from one state to the next. With the graph-valued state space specified, we now define the admissible action set for each state. \\

\begin{definition}[Graph Action Space]
\label{def:graph_action_space}
For any state \(s=(x,q)\), an action corresponds to selecting a neighboring node reachable from the agent's  node \(q\). Thus, the admissible action set is $\A(s):=\{v\in V:(q,v)\in E\}.$
Equivalently, each action may be identified with an admissible edge \((q,v)\in E\).
\end{definition}

We now specify the stochastic dynamics and objective of the graph-based MDP. \\

\begin{definition}[Graph Transition Kernel]
\label{def:graph_transition_kernel}
For any state \(s=(x,q)\in\mathcal{S}\) and action \(a\in \A(s)\), the kernel \(p(s'\mid s,a)\) denotes the probability of transitioning to a successor state \(s'=(x',q')\in\mathcal{S}\). This kernel characterizes the joint evolution of the graph valuation \(x\) and the agent location \(q\).
\end{definition}

There is a directed edge \(s\xrightarrow{a}s'\) in the state-transition graph whenever \(p(s'\mid s,a)>0\). Similarly, we define the reward structure associated with the graph-based MDP. \\

\begin{definition}[Graph-Based Reward Function]
\label{def:graph_reward_function}
The graph-based reward function is a possibly time-dependent mean reward map
$
r_t:\{(s,a):s\in\mathcal{S},\,a\in \A(s)\}\to\reals.
$
For a state $s=(x,q)\in\mathcal{S}$ and an action $a\in \A(s)$, the scalar $r_t(s,a)$ denotes the expected reward obtained at time $t$.
\end{definition}

Definitions~\ref{def:graph_state_space}--\ref{def:graph_reward_function} define the full graph-based MDP tuple. With a slight abuse of notation, we use the same representation $\mdp$ for this graph-based MDP, since the underlying MDP components remain unchanged, i.e., the representation is now interpreted through the graph structure. Next, we focus on group symmetries as a special type of MDP symmetry. \\

\begin{definition}[Group Symmetry of a Graph-Based MDP]
\label{def:graph_mdp_group_symmetry}
Let $M$ be a graph-based MDP whose state and action spaces are induced by graph $\G=(V,E)$. A group $G$ is a symmetry group of $M$ if each $g\in G$ acts on $V$ in a way that preserves the graph connectivity and induces maps on states and actions such that, for every state $s=(x,q)$ and admissible action $a\in\A(s)$,
$g\cdot s=(g\cdot x,g\cdot q)\in\mathcal{S}$
and
$g\cdot a\in\A(g\cdot s)$.
Moreover, the action preserves the MDP rewards and transition probabilities: for all successor states $c\in\mathcal{S}$,
$r_t(g\cdot s,g\cdot a)=r_t(s,a)$
and
$p(g\cdot c\mid g\cdot s,g\cdot a)=p(c\mid s,a)$.
\end{definition}

This is the standard notion of a group-invariant MDP~\cite{wangmathrm}, specialized here to graph-based MDPs. For further background, we refer the reader to Appendix~\ref{app:groups}. We finally define the orbit and stabilizer associated with a group action, which we will use to construct the quotient MDP. \\

\begin{definition}[Orbit and Stabilizer] \label{def:orbit-stabilizer} Let a finite group $G$ act on a set $\Omega$. For $\omega \in \Omega$, the \emph{orbit} and \emph{stabilizer} of $\omega$ are $\mathcal{G}_\omega \coloneqq \{g \cdot \omega : g \in G\}$ and $G_\omega \coloneqq \{g \in G : g \cdot \omega = \omega\} \leq G$, respectively. The orbit--stabilizer theorem~\citep[Thm.~5.1]{dummit2004abstract} gives $|\mathcal{G}_\omega| = |G|/|G_\omega|$. \\ 
\end{definition}

\begin{proposition}[Quotient MDP]
\label{prop:quotient-size}
Let a finite group $G$ act on the state space $\s$ and action space $\A$ of $M$. For $(s,a)\in\s\times\A$, let
$\mathrm{Succ}(s,a)=\{s'\in\s:p(s'\mid s,a)>0\}$
denote the set of possible successor states, and define
$C=\max_{s,a}|\mathrm{Succ}(s,a)|\leq|\s|$.
The quotient sizes are determined by the corresponding orbit sizes:
$|\bar{\s}|=|\s|/|\mathcal{G}_{\s}|$, $|\bar{\A}|=|\A|/|\mathcal{G}_{\A}|$, and $\bar C=C/|\mathcal{G}_C|$,
where $\mathcal{G}_{\s}$, $\mathcal{G}_{\A}$, and $\mathcal{G}_C$ denote the orbits of representative states, actions, and successor states, respectively.
\end{proposition}

\begin{proof}
By Definition~\ref{def:orbit-stabilizer}, quotienting identifies elements in the same orbit. Thus, the numbers of distinct state, action, and successor representatives are reduced by the corresponding orbit sizes.
\end{proof}

\subsection{MDP Homomorphisms}
The graph-based formulation above makes explicit the structure underlying the MDP. In particular, the topology of \(\G=(V,E)\) determines which actions are admissible from each state, while the graph valuation \(x:V\to X\) determines the configuration of local features across the environment. This structure allows us to compare states not only as distinct elements of \(\mathcal{S}\), but also according to whether they represent the same configuration up to a symmetry of the underlying graph.

For symmetric environments, this distinction is crucial. Two states \(s=(x,q)\) and \(s'=(x',q')\) may be different elements of \(\mathcal{S}=X^V\times V\), while still being structurally equivalent under a symmetry of \(\G\). A standard tabular MDP treats such states as unrelated, even when they induce identical rewards and transition behavior after relabeling the graph. This redundancy is best captured by MDP homomorphisms, which are structure-preserving maps from an original MDP to an abstract MDP~\cite{van2020mdp}, as conceptually visualized in Fig.~\ref{fig:quotient}. \\

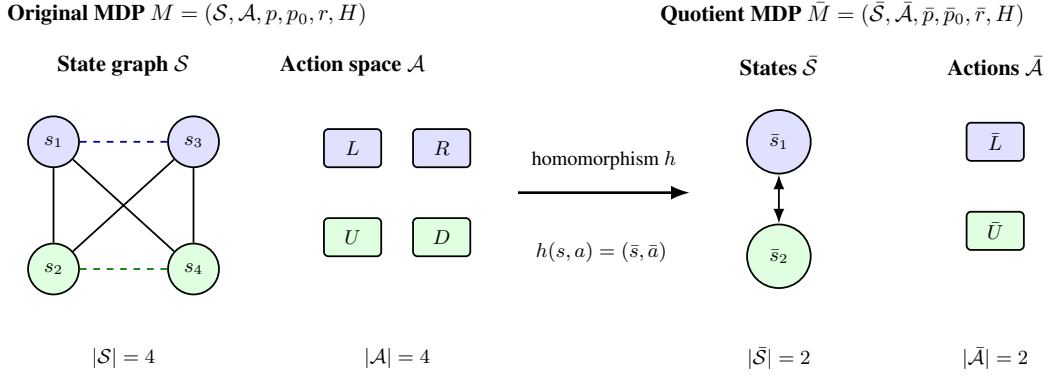
\begin{figure}[t]
\centering
\resizebox{\linewidth}{!}{%
\begin{tikzpicture}[
    >=Latex,
    font=\small,
    state/.style={
        circle,
        draw,
        thick,
        minimum size=8mm,
        inner sep=0pt
    },
    qstate/.style={
        circle,
        draw,
        thick,
        minimum size=10mm,
        inner sep=0pt
    },
    act/.style={
        draw,
        rounded corners=2pt,
        thick,
        minimum width=9mm,
        minimum height=6mm,
        inner sep=2pt
    },
    orbit1/.style={fill=blue!12},
    orbit2/.style={fill=green!12},
    maparrow/.style={->, very thick},
    edge1/.style={thick, blue!60!black},
    edge2/.style={thick, green!50!black}
]

\node[font=\bfseries] at (0.7,4.8) {Original MDP $\mdp$};

\node[font=\bfseries] at (-0.3,4.0) {State graph \(\s\)};

\node[state, orbit1] (s1) at (-1.4,2.8) {$s_1$};
\node[state, orbit1] (s3) at ( 0.8,2.8) {$s_3$};
\node[state, orbit2] (s2) at (-1.4,0.8) {$s_2$};
\node[state, orbit2] (s4) at ( 0.8,0.8) {$s_4$};

\draw[thick] (s1) -- (s2) -- (s3) -- (s4) -- (s1);

\draw[edge1, dashed] (s1) -- (s3);
\draw[edge2, dashed] (s2) -- (s4);


\node at (-0.3,-0.6) {\(|\s|=4\)};

\node[font=\bfseries] at (3.3,4.0) {Action space \(\A\)};

\node[act, orbit1] (L) at (3.3,2.7) {$L$};
\node[act, orbit1] (R) at (4.7,2.7) {$R$};
\node[act, orbit2] (U) at (3.3,1.3) {$U$};
\node[act, orbit2] (D) at (4.7,1.3) {$D$};


\node at (4.0,-0.6) {\(|\A|=4\)};

\draw[maparrow] (5.9,2.0) -- (8.6,2.0);
\node[above] at (7.25,2.25) {homomorphism \(h\)};
\node[below, align=center] at (7.25,1.35) {%
\(h(s,a)=(\bar{s},\bar{a})\)
};

\node[font=\bfseries] at (11.0,4.8) {Quotient MDP $\mdpb$};

\node[font=\bfseries] at (10.0,4.0) {States \(\bar{\s}\)};

\node[qstate, orbit1] (qs1) at (10.0,2.8) {\(\bar{s}_1\)};
\node[qstate, orbit2] (qs2) at (10.0,1.0) {\(\bar{s}_2\)};

\draw[thick, <->] (qs1) -- (qs2);

\node at (10.0,-0.6) {\(|\bar{\s}|=2\)};

\node[font=\bfseries] at (13.4,4.0) {Actions \(\bar{\A}\)};

\node[act, orbit1] (qL) at (13.4,2.8) {\(\bar{L}\)};
\node[act, orbit2] (qU) at (13.4,1.4) {\(\bar{U}\)};

\node at (13.4,-0.6) {\(|\bar{\A}|=2\)};

\end{tikzpicture}
}
\caption{Example where homomorphism $h$ maps the original MDP $M$ to the quotient MDP $\bar{M}$.
}

\label{fig:quotient}
\end{figure}

\begin{definition}[MDP Homomorphism]
\label{def:mdp-homomorphism}
Let $\mdp$ be an MDP as in Definition~\ref{def:mdp}, and let
$\mdpb$ be an abstract MDP. An MDP homomorphism from $M$ to $\bar{M}$ is a tuple
$h = (\sigma,\{\alpha_s\}_{s\in\s})$, where
$\sigma:\s\to\bar{\s}$ is a surjective state map and, for each $s\in\s$,
$\alpha_s:\A(s)\to\bar{\A}(\sigma(s))$ is a surjective action map. These maps satisfy, for all
$s,s'\in\s$ and $a\in\A(s)$,
\begin{align*}
\bar{r}\bigl(\sigma(s),\alpha_s(a)\bigr)
&= r(s,a), \qquad 
\bar{p}\bigl(\sigma(s') \mid \sigma(s),\alpha_s(a)\bigr)
= \sum_{s''\in\sigma^{-1}(\sigma(s'))} p(s''\mid s,a).
\end{align*}
\end{definition}
An exact MDP homomorphism induces a model-equivalent abstraction~\cite{li2006towards}: state-action pairs that are mapped to the same abstract state-action pair have identical reward and aggregated transition behavior. Given a homomorphism $h=(\sigma,\{\alpha_s\}_{s\in\mathcal{S}})$, two state-action pairs $(s,a)$ and $(s',a')$ are $h$-equivalent if
$\sigma(s)=\sigma(s')$ and $\alpha_s(a)=\alpha_{s'}(a').$
In this work, we use the group-symmetry setting introduced in Definition~\ref{def:graph_mdp_group_symmetry} as the relevant special case: the equivalence classes are the group orbits, and the resulting abstract model is the quotient MDP characterized in Proposition~\ref{prop:quotient-size}.  




\section{Equivariant Sample Complexity}
\label{sec:sample}

We consider the episodic fixed-horizon graph-based quotient MDP $\bar{M}$. The agent interacts with the MDP over episodes of length $H$, and we study the number of episodes in which the executed policy is $\epsilon$-suboptimal with probability at least $1-\delta$. We quantify the reduction in sample complexity bounds relative to the original MDP $M$. For details on the episodic PAC framework, we refer to Appendix~\ref{app:pac}.

\subsection{Equivariant Sample Complexity Bounds}

We are now ready to incorporate the group symmetry into the upper and lower bounds of the sample-complexity. 
We note that $\tilde{\mathcal{O}}(\cdot)$ suppresses logarithmic factors in $H$, $|\s|$, $|\A|$, $|G|$, $1/\epsilon$, and $1/\delta$. This corresponds to a Uniform-PAC guarantee, i.e., the bound is polynomial in all of its arguments \cite{dann2017unifying}.
Our strategy is to apply the standard PAC bound to the quotient MDP. 
By the homomorphism property, an optimal policy in $\bar{M}$ lifts to an optimal policy in $M$. \\

\begin{lemma}
\label{lem:homomorphism-sample-complexity}
Let $\mdp$ be an MDP and let $h$ be the mapping that induces an MDP homomorphism from $M$ to $\bar M$ as in Definition~\ref{def:mdp-homomorphism}, where $\mdpb$. Applying the result of \cite[Thm. 1]{dann2015sample} to $\bar M$  yields the following sample-complexity upper bound
$$ \tilde{\mathcal{O}}\!\left(\frac{H^2 \bar C|\bar{\s}\times\bar{\A}|}{\epsilon^2}\ln\frac{1}{\delta}\right).$$
\end{lemma}

\begin{proof}
By Definition~\ref{def:mdp-homomorphism}, the mapping $h$ induces the quotient MDP $\bar{M}$, whose state and action spaces are $\bar{\s}$ and $\bar{\A}$ and whose successor-state parameter is $\bar C$. Substituting these quantities into the bound of \cite{dann2015sample}, namely replacing $\s,\A,C$ by $\bar{\s},\bar{\A},\bar C$, directly yields the result.
\end{proof}

Lemma~\ref{lem:homomorphism-sample-complexity} applies the standard PAC sample-complexity bound directly to the quotient MDP. Theorem~\ref{th:upper} uses this substitution to quantify the resulting reduction relative to the original MDP. \\

\begin{theorem}[Upper Bound]
\label{th:upper}
Let $\mdp$ be an episodic fixed-horizon MDP and let $\epsilon \in (0,1)$ and $\delta \in (0,1)$. Then, with probability at least $1-\delta$, the number of episodes satisfying $
R^{*} - R^{\pi^k}
=
V_{1:H}^*(s_0) - V_{1:H}^{\pi^k}(s_0)
> \epsilon
$
is reduced by a factor of $|\mathcal{G}_{\s}|\,|\mathcal{G}_{\A}|\,|\mathcal{G}_C|$ relative to the corresponding bound for the original MDP, yielding the upper bound
$$
\tilde{\mathcal{O}}\!\left(
\frac{
H^2 C |\s \times \A|
}{
|\mathcal{G}_{\s}|\,|\mathcal{G}_{\A}|\,|\mathcal{G}_C|\,\epsilon^2
}
\ln \frac{1}{\delta}
\right).
$$
\end{theorem}

\begin{proof}
By Lemma~\ref{lem:homomorphism-sample-complexity}, the sample complexity is bounded in terms of $\bar{\s}$, $\bar{\A}$, and $\bar C$. Per Proposition~\ref{prop:quotient-size}, the equivalence classes on the quotient spaces satisfy $|\bar{\s}|=|\s|/|\mathcal{G}_{\s}|$ and $|\bar{\A}|=|\A|/|\mathcal{G}_{\A}|$. Similarly, the successor-state complexity satisfies $\bar C=C/|\mathcal{G}_C|$. Substituting these quantities yields the complexity reduction  $|\mathcal{G}_{\s}|\,|\mathcal{G}_{\A}|\,|\mathcal{G}_C|$.
\end{proof}


The utility of Theorem~\ref{th:upper} is best illustrated by the structural compression observed in the quotient MDP (Fig.~\ref{fig:quotient}). 
While the original problem is governed by a complexity factor of $C|S \times A| = 64$, the symmetry-aware mapping reduces this to $\bar{C}|\bar{S} \times \bar{A}| = 8$. This $2^3$ reduction is not merely a numerical artifact of this example, but a direct manifestation of the abstract bounds established in Theorem~\ref{th:upper}. 
The following corollary makes this reduction explicit for a cyclic ring (Fig.~\ref{fig:headline}). \\

\begin{corollary}[Ring Sample Complexity]
\label{cor:ring}
For a ring domain with $|G|$ cyclic rotations, the symmetry group $G$ acts freely on states, while actions and the local successor branching factor are preserved. Hence, Theorem~\ref{th:upper} yields
$
\tilde{\mathcal{O}}\!\left(
\frac{
H^2 C |\s \times \A|
}{
|G|\epsilon^2
}
\ln \frac{1}{\delta}
\right).
$
\end{corollary}
\begin{proof}
Quotienting reduces the number of states by $|G|$, while the action space $\A$ and the successor-state complexity $C$ is preserved.
\end{proof}

We next establish a matching lower bound. Here, $\tilde{\Omega}(\cdot)$ hides logarithmic factors, analogously to $\tilde{\mathcal{O}}(\cdot)$.\\

\begin{theorem}[Lower Bound]
\label{th:lower}
Let $\mdp$ be an episodic fixed-horizon MDP and let $\epsilon \in (0,1)$ and $\delta \in (0,1)$. Then, with probability at least $1-\delta$, the number of episodes satisfying $
R^{*} - R^{\pi^k}
=
V_{1:H}^*(s_0) - V_{1:H}^{\pi^k}(s_0)
> \epsilon
$
has a lower bound reduced by a factor of $|\mathcal{G}_{\s}|\,|\mathcal{G}_{\A}|$ relative to the corresponding bound for the original MDP, yielding the lower bound
$$
\tilde{\Omega}\!\left(
\frac{H^2 |\s \times \A|}{{|\mathcal{G}_{\s}||\mathcal{G}_\A|} \epsilon^2}
\ln \frac{1}{\tilde{\delta}}
\right),
$$
where $\tilde{\delta}=\delta+c$ for an absolute constant $c>0$.
\end{theorem}
\begin{proof}
By the MDP homomorphism, the reduced MDP has state space $\bar{\s}$ and action space $\bar{\A}$. The equivalence classes on the quotient spaces satisfy $|\bar{\s}|=|\s|/|\mathcal{G}_{\s}|$ and $|\bar{\A}|=|\A|/|\mathcal{G}_{\A}|$. Substituting yields the complexity reduction $|\mathcal{G}_{\s}||\mathcal{G}_\A|$.
\end{proof}

The theoretical gap illustrated in Fig.~\ref{fig:reduction} underscores the improvement in efficiency afforded by exploiting problem symmetries. For this graph-based MDP, where \(\mathcal{S}\) denotes the original state space before quotienting and \(|\mathcal{S}|=10^5\), the symmetry-aware bounds reflect a substantial (potentially five orders of magnitude) reduction in the number of episodes required to achieve convergence to 95\% optimality. This performance gap becomes more pronounced as the problem scale increases. \\

\begin{figure}[t!]
    \centering
    \includegraphics[width=0.6\textwidth]{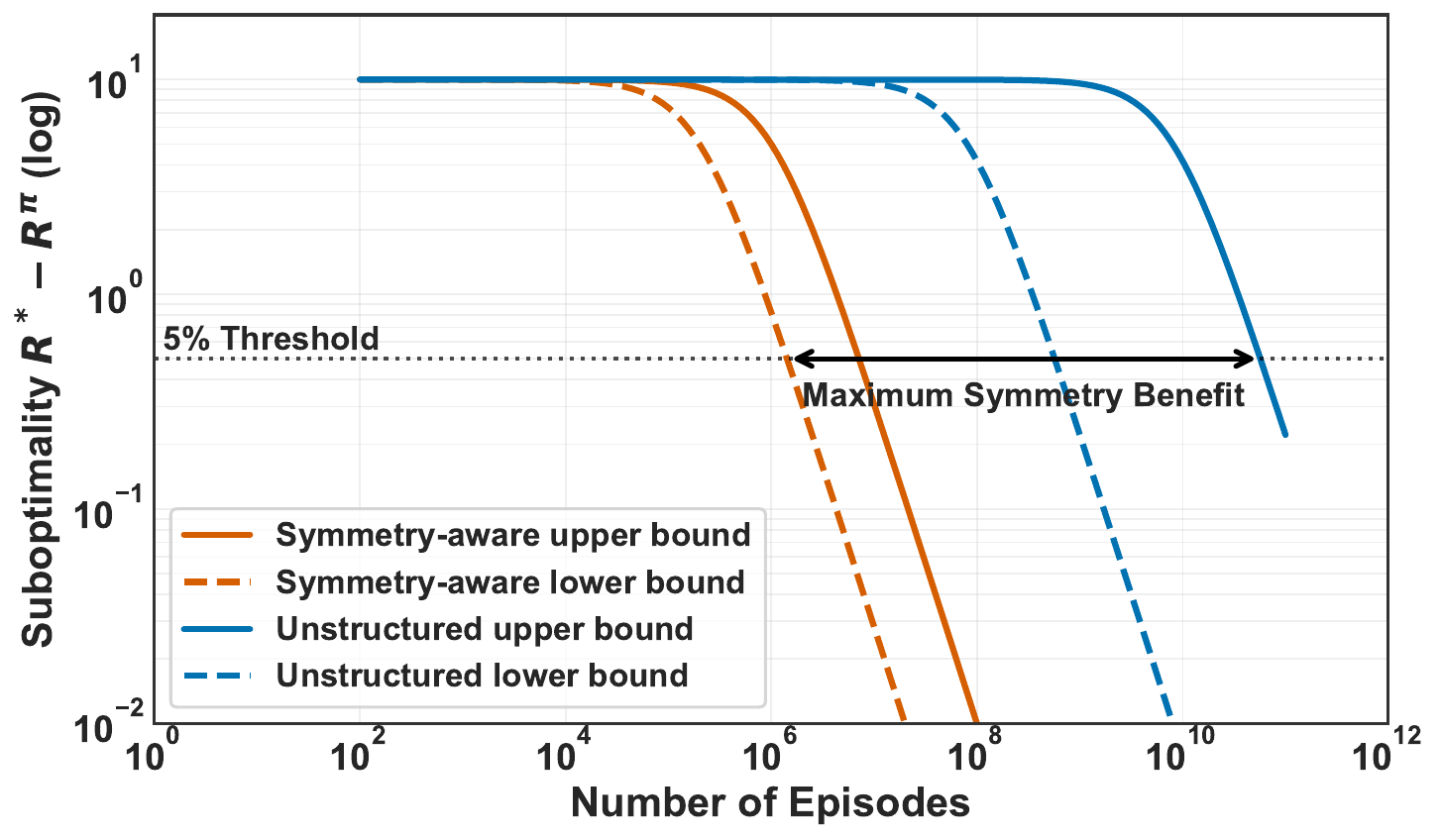}
    \caption{Theoretical comparison between unstructured and symmetry-aware sample complexity for a graph-based MDP with $|\mathcal{S}|=10^5$ states before quotienting by symmetry. For the ring graph, the symmetry group has full state-size symmetry, $|G|=10^5$, yielding a potential five-order-of-magnitude reduction in the number of episodes required to reach $95\%$ optimality with high probability.}
    \label{fig:reduction}
\end{figure}

\subsection{Continuous-Space Bounds}

We now extend the sample-complexity analysis to continuous domains by replacing finite cardinalities with covering numbers. In particular, under suitable regularity assumptions, the continuous MDP can be discretized by an equivariant $\alpha$-net that preserves the symmetry structure. The resulting quotient MDP then inherits the sample-complexity reduction of the discrete analysis, with the cardinalities of the quotient state and action spaces replaced by their corresponding covering
numbers. \\

\begin{theorem}[Continuous-Space Upper Bound]
\label{th:continuous-sample-complexity}
Let $(\s,\rho_{\s})$ and $(\A,\rho_{\A})$ be compact metric spaces for a finite-horizon MDP with $L_r$-Lipschitz rewards and $L_p$-Lipschitz transitions in the Wasserstein metric. Assume a group $G$ acts freely and isometrically on $\s$ and $\A$, and that the MDP is $G$-invariant. Then, for $\epsilon,\delta\in(0,1)$ and an equivariant discretization with resolution $\alpha=\alpha(\epsilon)\propto\epsilon/(HL_V)$, with probability at least $1-\delta$, the number of $\epsilon$-suboptimal episodes is upper bounded by
\[
\tilde{\mathcal{O}}\!\left(
\frac{
H^2 \bar C_{\alpha}
N(\s/G,\alpha)
N(\A/G,\alpha)
}{
\epsilon^2
}
\ln\frac{1}{\delta}
\right),
\]
where $\bar C_{\alpha}$ is the successor-state complexity of the discretized quotient MDP and $L_V$ is determined by the reward and transition Lipschitz constants.
\end{theorem}

\begin{proof}
Let $\alpha>0$ and construct equivariant $\alpha$-nets of $\s$ and $\A$. Since the group action is isometric, the discretization preserves the $G$-symmetry and induces a quotient discretized MDP. Under the Lipschitz assumptions, the value function is Lipschitz with constant $L_V$. Choosing
\[
\alpha(\epsilon)=\frac{\epsilon}{2c_0HL_V},
\]
for an absolute constant $c_0$, bounds the discretization error by $\epsilon/2$. The resulting quotient MDP has state and action cardinalities bounded by $N(\s/G,\alpha(\epsilon))$ and $N(\A/G,\alpha(\epsilon))$, respectively. Applying Theorem~\ref{th:upper} to this quotient MDP gives
\[
\tilde{\mathcal{O}}\!\left(
\frac{
H^2 \bar C_{\alpha(\epsilon)}
N(\s/G,\alpha(\epsilon))
N(\A/G,\alpha(\epsilon))
}{
\epsilon^2
}
\ln\frac{1}{\delta}
\right).
\]
Finally, lifting the resulting policy to the original continuous MDP through the homomorphism introduces at most $\epsilon/2$ discretization error, yielding an $\epsilon$-optimal policy and the stated bound.
\end{proof}

In the next section, we empirically evaluate the theoretical results and their implications across discrete and continuous domains.

\section{Experiments}

\begin{figure}[t!]
    \centering
    \begin{tabular}{@{}c@{\hspace{0.02\textwidth}}c@{\hspace{0.02\textwidth}}c@{}}
        \begin{subfigure}[b]{0.31\textwidth}
            \centering
            \includegraphics[width=\linewidth]{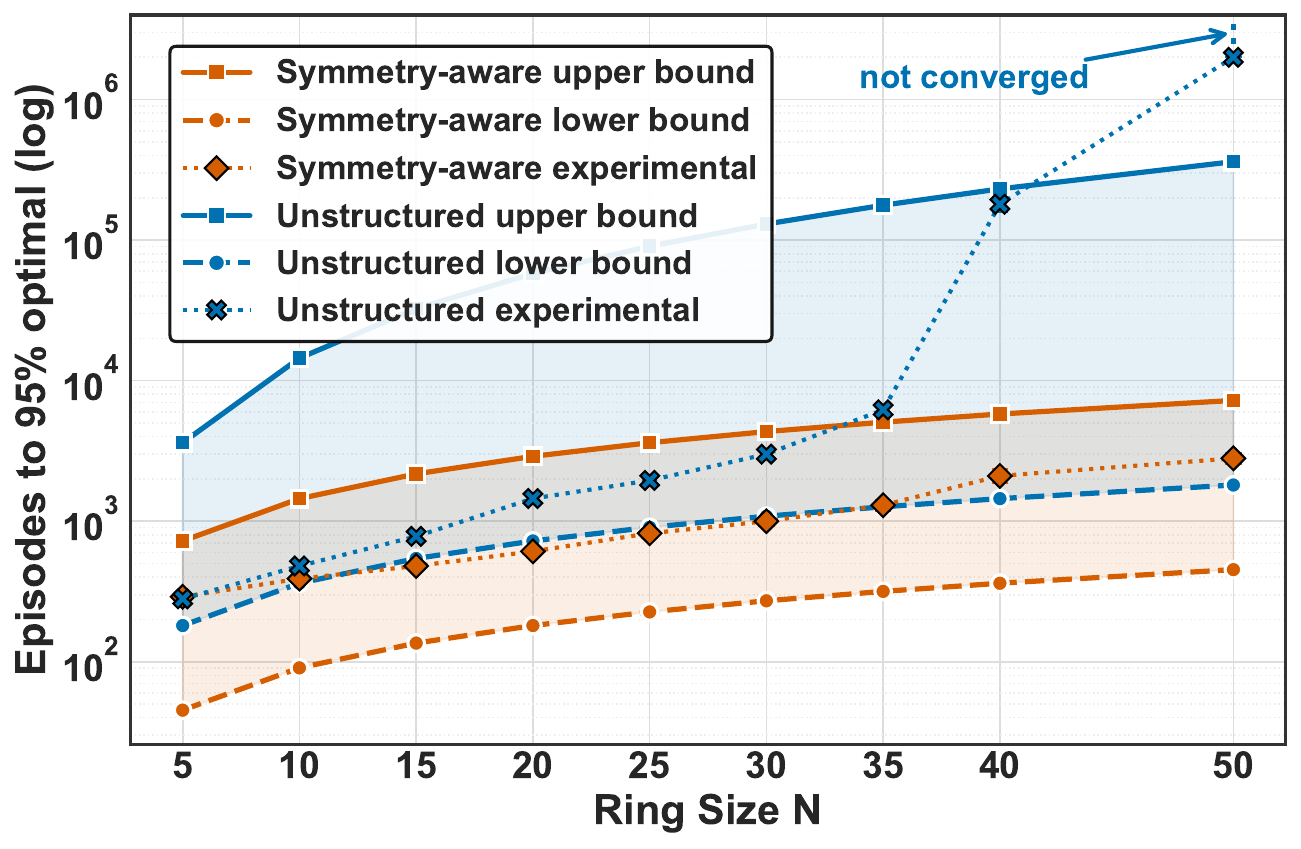}
            \caption{Cyclic ring experiment (1D).}
            \label{fig:sample_complexity-1D}
        \end{subfigure}
        &
        \begin{subfigure}[b]{0.33\textwidth}
            \centering
            \includegraphics[width=\linewidth]{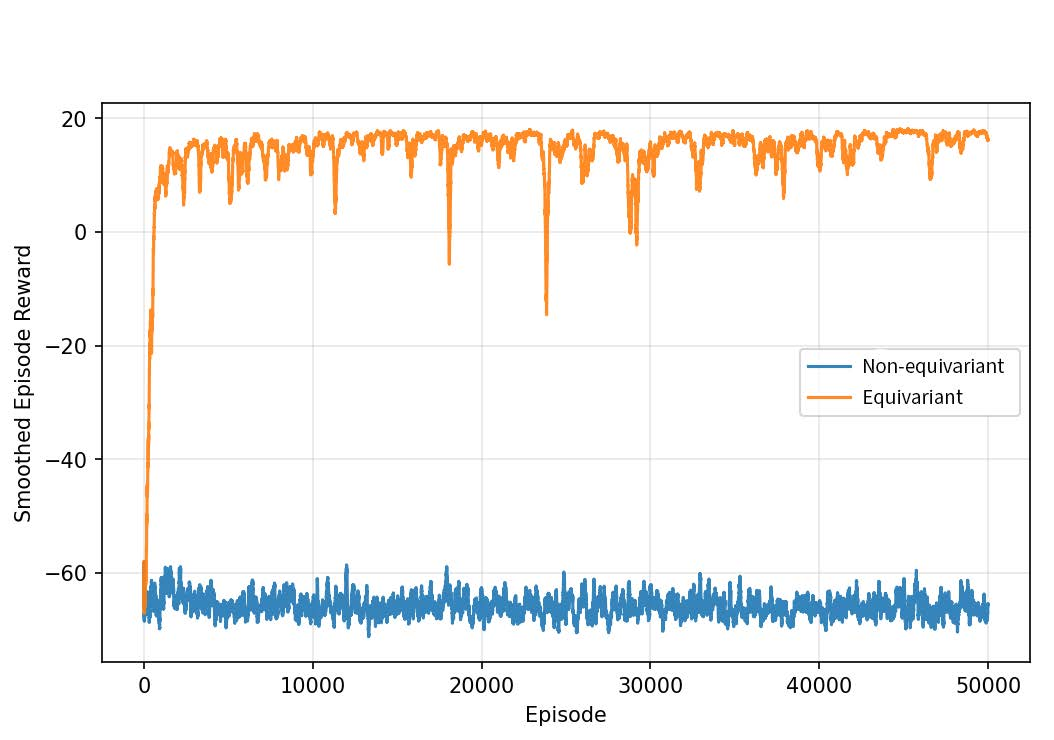}
            \caption{Cyclic ring experiment (3D).}
            \label{fig:sample_complexity-3D}
        \end{subfigure}
        &
        \hspace{-2mm}
        \begin{subfigure}[b]{0.37\textwidth}
            \centering
            \includegraphics[width=\linewidth]{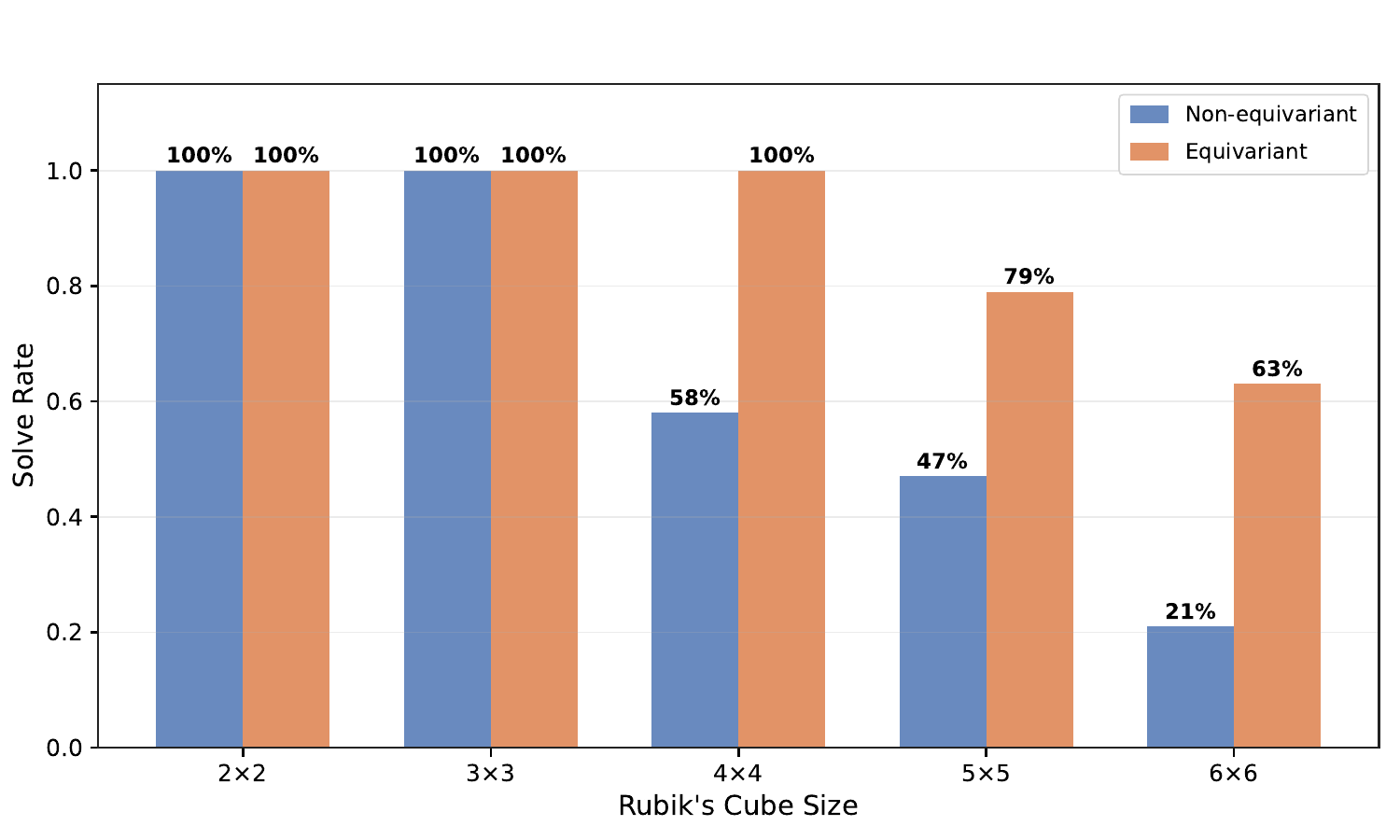}
            \caption{Rubik's cube experiment.}
            \label{fig:rubik}
        \end{subfigure}
    \end{tabular}

    \caption{Performance across the 1D cyclic ring, 3D cyclic ring, and Rubik's cube experiments. Symmetry-aware learning reduces the number of episodes needed for better performance.}
    \label{fig:experiments}
\end{figure}

We evaluate the practical implications of the sample complexity bounds on a set of representative environments with symmetry. Across all experiments, the baseline and symmetry-aware variants are compared under identical training settings, including network architecture, optimizer, and hyperparameter choices. We consider five tasks of increasing structural complexity: 1) a cyclic environment illustrating basic group symmetry, 2) a toroidal environment with three coupled rings illustrating product cyclic structure, 3) a structured combinatorial domain based on a Rubik's game, 4) a block-stacking robotics task and 5) a nut-assembly robotics task, both from~\cite{mandlekar2023mimicgen}. Additional setup and implementation details are provided in Appendix~\ref{app:setup}.

\subsection{Results}

\paragraph{Structured Domains.} We first consider ring-structured graphs with $N$ nodes. In the 1D ring, the symmetry group is the cyclic group $G=\mathbb{Z}/N\mathbb{Z}$ of rotations, so $|G|=N$. In the 3D ring, the graph is the product of three cyclic rings with symmetry group $G=(\mathbb{Z}/N\mathbb{Z})^3$, so $|G|=N^3$. Figs.~\ref{fig:sample_complexity-1D} and~\ref{fig:sample_complexity-3D} show the theoretical bounds and empirical learning curves, respectively. The equivariant method reaches higher reward within fewer episodes and exhibits substantially better sample efficiency than the non-equivariant baseline. This effect is more pronounced in 3D, where the non-equivariant baseline incurs high loss and fails to converge within the training budget. For the unstructured baseline at $N=50$ (Fig.~\ref{fig:sample_complexity-1D}), the empirical sample complexity exceeds the plotted theoretical range and the method fails to reach the required performance threshold within the training budget. A similar trend appears in the Rubik's cube experiments in Fig.~\ref{fig:rubik}. Canonicalization maps symmetry-equivalent cube states to a common representative, reducing the number of configurations the learner must distinguish. The symmetry-aware method maintains higher solve rates as the problem size increases, whereas the non-equivariant baseline degrades more rapidly, demonstrating improved information sharing across equivalent configurations.

\begin{table*}[h!]
\centering
\begin{minipage}[t]{0.48\textwidth}
\centering
\caption{Convergence for the 1D cyclic ring over 10 independent random seeds (mean $\pm$ SE).}
\label{tab:ring-multiseed}
\resizebox{0.90\linewidth}{!}{%
\begin{tabular}{c|cc}
\toprule
$N$ & Non-equiv. & Equiv. \\
\midrule
5  & $\mathbf{280.5 \pm 20.9}$ & $290.1 \pm 13.0$ \\
10 & $484.5 \pm 22.9$ & $\mathbf{405.1 \pm 18.0}$ \\
15 & $788.6 \pm 33.0$ & $\mathbf{493.1 \pm 15.6}$ \\
20 & $1{,}444.1 \pm 137.0$ & $\mathbf{617.1 \pm 28.8}$ \\
25 & $1{,}966.1 \pm 129.4$ & $\mathbf{829.1 \pm 63.9}$ \\
30 & $3{,}016.3 \pm 158.0$ & $\mathbf{997.0 \pm 71.0}$ \\
35 & $6{,}206.6 \pm 635.4$ & $\mathbf{1{,}350.9 \pm 144.2}$ \\
40 & $177{,}334.3 \pm 21{,}201.9$ & $\mathbf{2{,}157.4 \pm 208.2}$ \\
\bottomrule
\end{tabular}
}
\end{minipage}
\hfill
\begin{minipage}[t]{0.48\textwidth}
\centering
\caption{Rubik's cube success rates over 10 independent random seeds (mean $\pm$ SE). }
\label{tab:rubik-multiseed}
\renewcommand{\arraystretch}{2.21}
\resizebox{0.98\linewidth}{!}{
\begin{tabular}{c|ccc}
\toprule
Cube size & Non-equiv. & Canonical & $O$-equiv. \\
\midrule
$2\times2$ & $\mathbf{100.0 \pm 0.0\%}$ & $\mathbf{100.0 \pm 0.0\%}$ & $\mathbf{100.0 \pm 0.0\%}$ \\
$3\times3$ & $98.7 \pm 1.5\%$ & $\mathbf{99.3 \pm 1.2\%}$ & $99.0 \pm 1.7\%$ \\
$4\times4$ & $61.7 \pm 4.0\%$ & $\mathbf{97.0 \pm 3.0\%}$ & $95.0 \pm 2.6\%$ \\
$5\times5$ & $42.7 \pm 4.5\%$ & $\mathbf{83.7 \pm 4.5\%}$ & $81.0 \pm 2.6\%$ \\
\bottomrule
\end{tabular}
}
\end{minipage}

\end{table*}

\paragraph{Ablation Study.}
For the structured domains, we further evaluate the robustness of the observed gains over 10 independent random seeds. In the 1D cyclic ring, the results in Table~\ref{tab:ring-multiseed} show that the advantage of equivariant learning grows with $N$: while both methods perform similarly for small rings, the gap increases substantially with problem size, with the non-equivariant method requiring $177{,}334.3\pm21{,}201.9$ steps at $N=40$ compared with $2{,}157.4\pm208.2$ for the equivariant method. The same behavior is observed in the 3D cyclic ring, where the final reward is $-67.38\pm1.12$ SE for the non-equivariant method and $16.96\pm0.32$ SE for the equivariant method. Finally, Table~\ref{tab:rubik-multiseed} compares the non-equivariant baseline, canonicalization, and an $O$-equivariant MDP-homomorphic network on the Rubik's cube task. The $O$-equivariant method uses an MDP-homomorphic network~\cite{van2020mdp}. Both symmetry-aware approaches remain effective as the cube size increases and substantially outperform the non-equivariant baseline, while canonicalization performs comparably to or slightly better than the network-level equivariant approach. Together, these results show that the gains from exploiting symmetry persist across random seeds, become more pronounced as problem complexity grows, and are not specific to a particular realization of equivariance.

\begin{figure}[t!]
    \centering

    \begin{subfigure}[b]{0.49\textwidth}
        \centering
        \includegraphics[width=\linewidth]{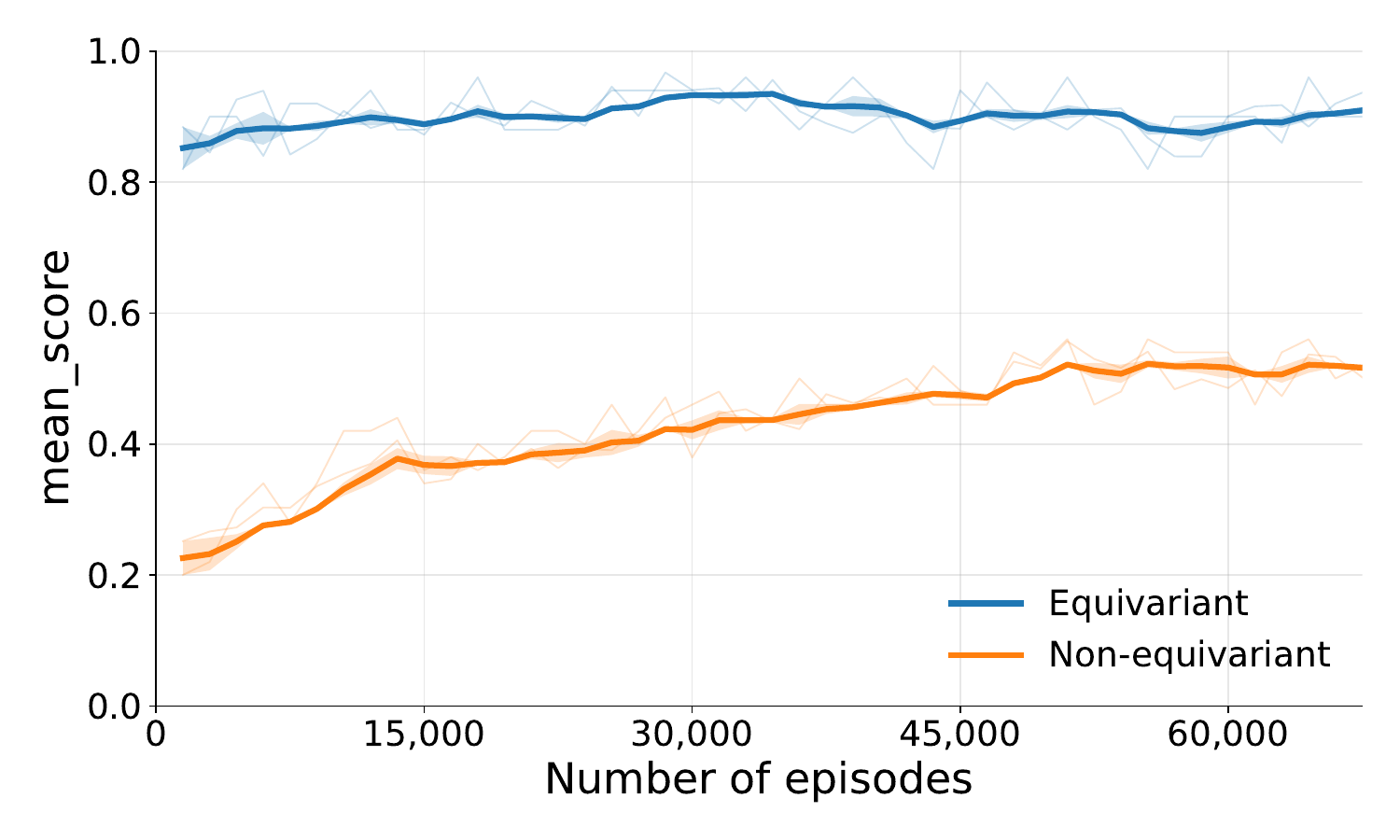}
        \caption{Block stacking benchmark.}
        \label{fig:stack}
    \end{subfigure}
    \hfill
    \begin{subfigure}[b]{0.49\textwidth}
        \centering
        \includegraphics[width=\linewidth]{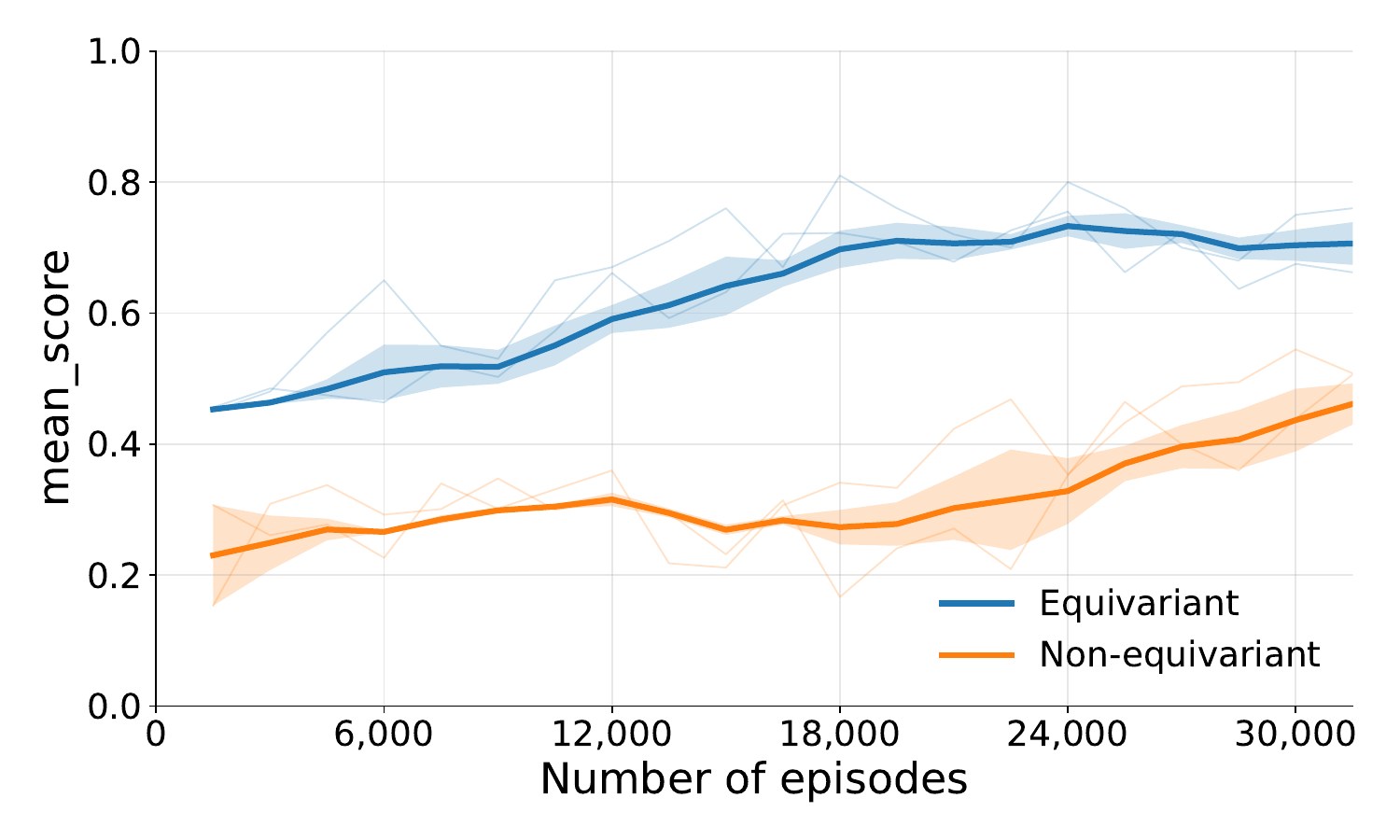}
        \caption{Nut assembly benchmark.}
        \label{fig:nut}
    \end{subfigure}

    \caption{Continuous-control manipulation benchmarks comparing symmetry-aware and non-equivariant learning across multiple random seeds.}
    \label{fig:rl-tasks}
\end{figure}

\paragraph{Continuous-Control Tasks.}
We additionally evaluate the framework on \emph{continuous-control} robotic manipulation benchmarks from~\cite{mandlekar2023mimicgen}. These tasks involve high-dimensional continuous observations and actions and contact-rich dynamics. In this setting, we model the relevant symmetry group as $\mathrm{SE}(3)$, the special Euclidean group in 3D. These experiments complement Theorem~\ref{th:continuous-sample-complexity}, which extends the symmetry-induced sample-complexity reduction to continuous state-action spaces through quotient covering numbers. In the block stacking benchmark (Fig.~\ref{fig:stack}), the agent must learn manipulation behavior that generalizes across symmetry-related object configurations. The non-equivariant baseline exhibits worse learning behavior, indicating that it must relearn similar behaviors across equivalent configurations. By contrast, the equivariant agent reaches useful performance levels earlier and achieves higher return, showing that the symmetry-aware representation improves learning efficiency. In the nut assembly benchmark (Fig.~\ref{fig:nut}), the task requires precise contact-rich control and alignment, making learning more sensitive to exploration and initialization. The same trend persists: the equivariant method consistently outperforms the non-equivariant baseline across training, with higher mean return. These results are consistent with Theorem~\ref{th:continuous-sample-complexity}, demonstrating that the benefits of symmetry-aware learning extend empirically to continuous-control domains.

\section{Conclusion}

This work studies how symmetry-induced structure changes the sample complexity of reinforcement learning. For finite MDPs, we formalize how graph-induced symmetries partition redundant state-action representations into orbits and show how the resulting quotient changes the effective size of the learning problem. We further extend this analysis to continuous state-action spaces, where symmetry reduces the effective complexity through the covering numbers of the corresponding quotient spaces. The resulting bounds demonstrate that symmetry-aware learning can improve sample efficiency in a way that scales with the amount of redundancy present in the environment. Our results show that integrating symmetry into the learning pipeline yields substantial gains in sample efficiency and performance, offering a principled path toward more data-efficient general-purpose robotics. Future work will extend this framework toward equivariant world models for and on-robot demonstrations.

\clearpage









\bibliography{cite}


\appendix

\section{Group-Invariant MDPs}
\label{app:groups}

We expand on the rationale for incorporating group symmetry, following the formulation of \cite{wangmathrm}. \\

\begin{definition}[{Group-Invariant MDP}]
  A group-invariant MDP is a tuple
$ M_G = (S, A, T, R, G),$
where $G$ is a group acting on both the state space $S$ and the action space $A$. For any $g \in G$, let $gs \in S$ denote the action of $g$ on state $s$, and $ga \in A$ denote the action of $g$ on action $a$. We require the following invariance properties:
\begin{enumerate}
  \item \emph{Reward invariance:} $R(s,a) = R(gs, ga)$ for all $s \in S, a \in A, g \in G$.
  \item \emph{Transition invariance:} $T(s,a) = s'$ implies $T(gs,ga) = g s'$ for all $s \in S, a \in A, g \in G$.
\end{enumerate}  
\end{definition}

\paragraph{Main result.}
Under these assumptions, the optimal solutions inherit group symmetry:

\begin{itemize}
  \item The optimal value function is group-invariant:
  \[
  V^*(gs) = V^*(s), \quad \forall s \in S, \; g \in G.
  \]
  \item The optimal Q-function is group-invariant:
  \[
  Q^*(s,a) = Q^*(gs,ga), \quad \forall (s,a) \in S \times A, \; g \in G.
  \]
  \item The optimal policy is group-equivariant:
  \[
  \pi^*(gs) = g \pi^*(s), \quad \forall s \in S, \; g \in G.
  \]
\end{itemize}

\paragraph{Proof sketch~\cite{wangmathrm}.}
 Consider the Bellman optimality equation for $Q^*$:
\[
Q^*(s,a) = R(s,a) + \gamma \max_{a'} Q^*(T(s,a), a').
\]
For the transformed pair $(gs,ga)$ we have
\[
Q^*(gs,ga) = R(gs,ga) + \gamma \max_{a'} Q^*(T(gs,ga), a').
\]
By reward and transition invariance, this reduces to
\[
Q^*(gs,ga) = R(s,a) + \gamma \max_{a'} Q^*(gT(s,a), ga').
\]
Re-indexing by $a'' = ga'$ and using the fact that $g$ permutes actions, this is identical to the equation for $Q^*(s,a)$. By uniqueness of Bellman solutions, $Q^*(gs,ga) = Q^*(s,a)$.

Now, since $V^*(s) = \max_a Q^*(s,a)$, it follows immediately that
\[
V^*(gs) = \max_a Q^*(gs,a) = \max_a Q^*(s,g^{-1}a) = V^*(s).
\]

Finally, the optimal policy is defined as
\[
\pi^*(s) = \arg\max_a Q^*(s,a).
\]
Using $Q^*$-invariance, we obtain
\[
\pi^*(gs) = \arg\max_a Q^*(gs,a) = g \arg\max_{a'} Q^*(s,a') = g\pi^*(s).
\]

Thus, for a group-invariant MDP, the optimal value function is invariant, the optimal $Q$-function is invariant, and the optimal policy is equivariant with respect to the group.

\section{PAC Framework}
\label{app:pac}

We outline the theoretical PAC implications of the episodic PAC framework in~\cite{dann2015sample}. The learner proceeds in phases indexed by $k$, where each phase is the interval between two model updates. At the start of phase $k$, the learner constructs visit counts, empirical transition estimates, and confidence sets over state-action configurations. 

Let $\mathcal{M}_k$ denote the corresponding confidence set of statistically plausible MDPs at phase $k$, constructed from empirical reward and transition estimates. The algorithm selects an optimistic model $\tilde{M}_k\in\mathcal{M}_k$, i.e., a model whose optimal value is maximal over the confidence set, computes a policy $\pi^k$, and executes it until the model is updated. The following PAC property is used to establish the sample-complexity bounds. \\

\begin{lemma}
With probability at least $1 - \delta/2$, the true MDP $M$ lies in $\mathcal{M}_k$ for all phases $k$.
\end{lemma}

\begin{proof}[Proof Sketch]
For a fixed phase $k$, state-action pair $(s,a)$, and successor state $s'\in\mathrm{Succ}(s,a)$, concentration inequalities imply that $p(s'\mid s,a)$ lies in the confidence interval around $\hat{p}_k(s'\mid s,a)$ with high probability \cite{maurer2009empirical}. A union bound over successor states, state-action pairs, and model updates yields that the true MDP belongs to $\mathcal{M}_k$ for all phases with probability at least $1-\delta/2$.
\end{proof}

Intuitively, the lemma gives the high-probability event on which the true MDP remains plausible throughout learning.  For more details on the PAC derivation, we refer to \cite{dann2015sample}.

\section{Experimental Setup}
\label{app:setup}

We provide the experimental details needed to reproduce the results in the main paper. Illustrations of the structured domains are shown in Fig.~\ref{fig:experiment-illustrations}. We first report the compute budget for each setting, then describe the environments, symmetry groups, canonicalization procedures, model architectures, and optimization hyperparameters used in the 1D and 3D cyclic games, Rubik's Cube, and Mimicgen manipulation experiments.

\subsection{Compute}

All RL experiments in this paper train on a single NVIDIA A100 40GB GPU;
the 1D cyclic-game DQN sweep runs on CPU. The numbers below are per-run
budgets; multiply by the number of $(N, \text{condition}, \text{seed})$
combinations to get sweep totals.

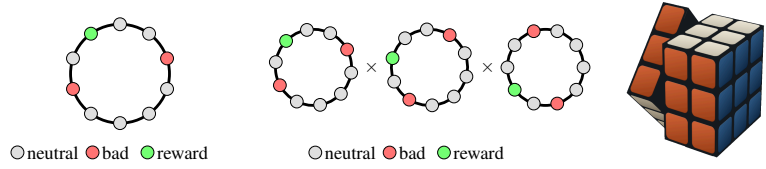
\begin{figure}[t!]
    \centering
    \hspace{-8mm}
    \begin{tabular}{@{}c@{\hspace{0.02\textwidth}}c@{\hspace{0.02\textwidth}}c@{}}
        \begin{subfigure}[b]{0.20\textwidth}
            \centering
            \resizebox{\linewidth}{!}{%
                \usetikzlibrary{calc}

\begin{tikzpicture}[
    >=Latex,
    font=\Large,
    neutraldot/.style={circle, fill=gray!25, draw=black, minimum size=10pt, inner sep=0pt},
    baddot/.style={circle, fill=red!55, draw=black, minimum size=10pt, inner sep=0pt},
    rewarddot/.style={circle, fill=green!55, draw=black, minimum size=10pt, inner sep=0pt},
    ringedge/.style={draw=black, line width=2.2pt}
]

\def\R{1.35}

\draw[ringedge, fill=white] (0,0) circle[radius=\R];

\foreach \i in {0,...,9} {
    \coordinate (s\i) at ({90 - 36*\i}:\R);
}

\foreach \i in {0,1,3,4,5,6,8} {
    \node[neutraldot] at (s\i) {};
}

\foreach \i in {2,7} {
    \node[baddot] at (s\i) {};
}

\node[rewarddot] at (s9) {};

\begin{scope}[shift={(-2.8,-2.15)}]
    \node[neutraldot] at (0,0) {};
    \node[right=4pt] at (0,0) {neutral};

    \node[baddot] at (2.05,0) {};
    \node[right=4pt] at (2.05,0) {bad};

    \node[rewarddot] at (3.55,0) {};
    \node[right=4pt] at (3.55,0) {reward};
\end{scope}

\end{tikzpicture}
            }
        \end{subfigure}
        &
        \hspace{3mm}
        \begin{subfigure}[b]{0.31\textwidth}
            \centering
            \resizebox{\linewidth}{!}{%
                \begin{tikzpicture}[
    >=Latex,
    font=\Large,
    neutraldot/.style={circle, fill=gray!25, draw=black, minimum size=10pt, inner sep=0pt},
    baddot/.style={circle, fill=red!55, draw=black, minimum size=10pt, inner sep=0pt},
    rewarddot/.style={circle, fill=green!55, draw=black, minimum size=10pt, inner sep=0pt},
    ringedge/.style={draw=black, line width=2.2pt}
]

\newcommand{\drawnewring}{%
    \def\R{1.05}

    \draw[ringedge, fill=white] (0,0) circle[radius=\R];

    \foreach \i in {0,...,9} {
        \coordinate (s\i) at ({90 - 36*\i}:\R);
    }

    \foreach \i in {0,1,3,4,5,6,8} {
        \node[neutraldot] at (s\i) {};
    }

    \foreach \i in {2,7} {
        \node[baddot] at (s\i) {};
    }

    \node[rewarddot] at (s9) {};
}

\begin{scope}[shift={(-3.2,0)}, rotate=10]
    \drawnewring
\end{scope}

\node[font=\Large] at (-1.6,0) {$\times$};

\begin{scope}[shift={(0,0)}, rotate=40]
    \drawnewring
\end{scope}

\node[font=\Large] at (1.6,0) {$\times$};

\begin{scope}[shift={(3.2,0)}, rotate=90]
    \drawnewring
\end{scope}

\begin{scope}[shift={(-3.15,-2.35)}]
    \node[neutraldot] at (0,0) {};
    \node[right=4pt] at (0,0) {neutral};

    \node[baddot] at (2.05,0) {};
    \node[right=4pt] at (2.05,0) {bad};

    \node[rewarddot] at (3.55,0) {};
    \node[right=4pt] at (3.55,0) {reward};
\end{scope}

\end{tikzpicture}
            }
        \end{subfigure}
        &
        \begin{subfigure}[b]{0.16\textwidth}
            \centering
            \includegraphics[width=\linewidth]{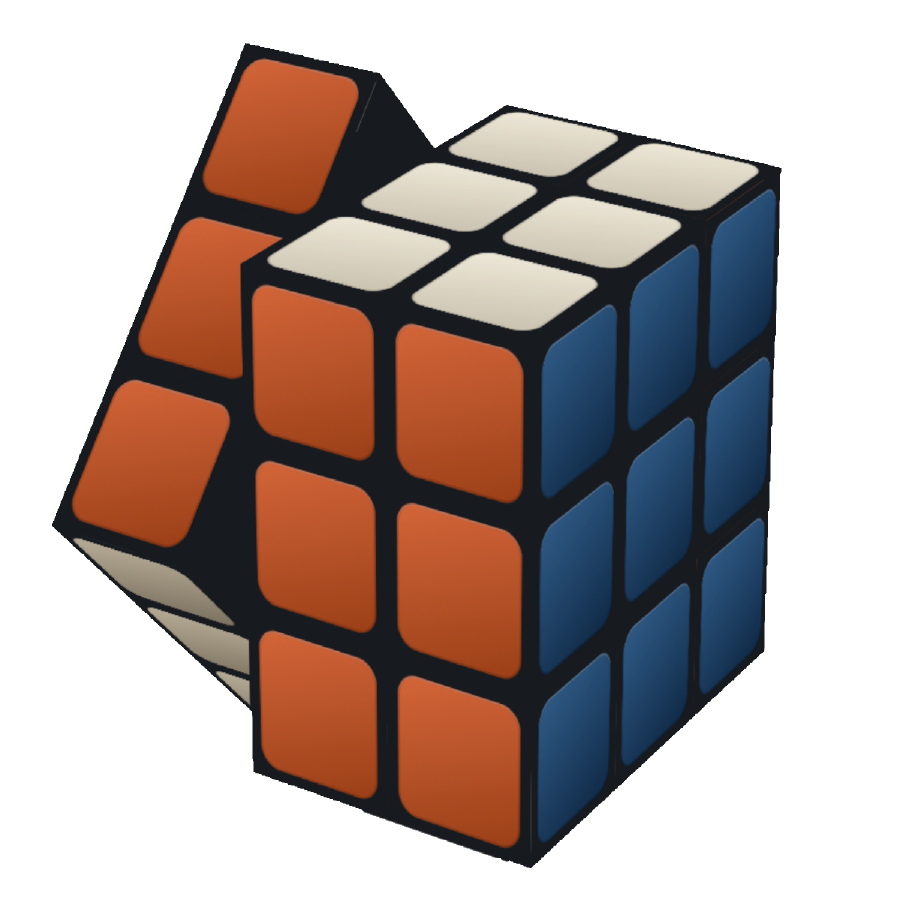}
        \end{subfigure}
    \end{tabular}

    \caption{Illustrations of the 1D cyclic ring, 3D cyclic ring, and Rubik's cube experiments.}
    \label{fig:experiment-illustrations}
\end{figure}

\paragraph{1D Cyclic Game.}
Each run ($10{,}000$ episodes) takes approximately $3$-$20$ CPU-hours,
scaling with $N$ via $\mathrm{max\_steps}=4N$.

\paragraph{3D Cyclic Game.}
Each run ($50{,}000$ episodes; one of $\{$DQN, PPO$\}$ $\times$
$\{$baseline, canonicalized$\}$) takes approximately $8$ A100-hours.

\paragraph{Rubik's Cube.}
Each $10{,}000$-epoch run takes approximately $20$ A100-hours; the
extended $N{=}5$ run at $50{,}000$ epochs scales accordingly.

\paragraph{Mimicgen PPO.}
Each run ($16$ parallel envs, $500$ episodes per training rollout,
$2\!\times\!10^{6}$ timesteps) takes approximately $76$ A100-hours.

\subsection{Experiments}

\subsubsection{1D Cyclic Game}
\label{app:exp:1d}

\paragraph{Environment.}
The 1D cyclic game is the single-axis restriction of the construction in
Section~\ref{app:exp:3d}: a ring $\mathcal{S}=\mathbb{Z}/N\mathbb{Z}$ of $N$
cells, each holding an item $b_c\in\{0,-1,E\}$ (safe, mine, end). At reset,
$b_c=-1$ with probability $p$ and $b_c=0$ otherwise; a single end-spot
$c^\star\in\mathcal{S}$ is sampled uniformly with reward $E>0$. The action
space is $\mathcal{A}=\{-2,-1,1,2\}$ ($|\mathcal{A}|=4$), with deterministic
cyclic transitions $s_{t+1}=(s_t+a_t)\bmod N$. Episodes end when the agent
reaches $c^\star$ or after $\mathrm{max\_steps}=4N$ steps. We fix $p=0.7$ and
$E=10$.

\paragraph{Symmetry group and canonicalization.}
The translation group $G=\mathbb{Z}/N\mathbb{Z}$ acts on $\mathcal{S}$ by
$(g\!\cdot\! b)_c = b_{c+g\,\mathrm{mod}\, N}$ and trivially on actions, so
the orbit-reduction factor is $|G|=N$. As in the 3D case, the canonical
environment rolls the items vector by $-c^\star$ on every \texttt{reset},
fixing the goal at the origin and shifting the agent's start to
$(-c^\star)\bmod N$. The policy observes the concatenation
$(b\,\Vert\,\mathrm{onehot}(\mathrm{pos}))\in\mathbb{R}^{2N}$.

\paragraph{Sweep over $N$.}
We sweep $N\in\{5,10,15,20,25,30\}$ to study how sample complexity scales
with state-space size, training the baseline and canonicalized variants
under the same seed at each $N$. All runs use seed $42$ and $10{,}000$
episodes per $(N,\text{condition})$ pair.

\paragraph{Algorithm.}
Both variants use the same MLP Q-network ($2N\!\to\!128\!\to\!128\!\to\!4$,
ReLU). We train DQN ($\gamma{=}0.99$, replay buffer $10{,}000$, batch $64$,
target update every $10$ episodes, $\epsilon$ from $1.0$ to $0.05$ with
decay $0.997$, Adam lr $10^{-3}$).

\subsubsection{3D Cyclic Game}
\label{app:exp:3d}

\paragraph{Environment.}
The 3D cyclic game is a discrete grid-world living on the 3-torus
$\mathcal{S}=(\mathbb{Z}/N\mathbb{Z})^3$. Each cell $c\in\mathcal{S}$ holds an item
$b_c\in\{0,-1,E\}$ (safe, mine, end), drawn at reset time as
$b_c=-1$ with probability $p$ and otherwise $b_c=0$, with a single end-spot
$c^\star\in\mathcal{S}$ placed uniformly at random and reward $E>0$ assigned to
that cell. The action space is $\mathcal{A}=\{-2,-1,1,2\}^3$ with $|\mathcal{A}|=64$,
and transitions are deterministic axis-wise cyclic shifts
$s_{t+1}=(s_t+a_t)\bmod N$. An episode terminates when the agent reaches $c^\star$ or
exceeds the step limit. We use $d=3$, $N=5$, $p=0.7$,
$E=20$, $\mathrm{max\_steps}=100$ for evaluation and $100$ during training.

\paragraph{Symmetry group.}
The translation group $G=(\mathbb{Z}/N\mathbb{Z})^3$ acts on $\mathcal{S}$ by
$(g\!\cdot\! b)_c = b_{c+g\,\mathrm{mod}\, N}$ and on actions by the trivial
representation; reward and dynamics are $G$-invariant. With end-spot position included
in the state, every game admits a unique transversal under $G$, so
$|\bar{\mathcal{S}}|=|\mathcal{S}|/|G|$ and the orbit-reduction factor is exactly
$|G|=N^3$ (e.g.\ $1.25\times 10^{5}$ for $N=50$).

\paragraph{Canonicalization.}
Symmetry is exploited \emph{at the environment level} rather than through an
equivariant network. At every \texttt{reset}, the canonical environment translates
the items grid by $-c^\star$ along each axis, fixing the goal at the origin and
rotating the agent's start position accordingly:
\[
b'\;=\;\textstyle\prod_{a=0}^{2}\mathrm{roll}(b,\,-c^\star_a,\,\text{axis}=a),
\qquad c^{\star\,\prime}=(0,0,0),
\qquad s_0' = (0-c^\star)\,\mathrm{mod}\,N.
\]

\paragraph{Algorithms.}
Both the baseline and canonicalized variants share the same architecture and
optimizer. The Q-network and actor-critic both use a 3-layer 3D-CNN encoder
($1\!\to\!16\!\to\!32\!\to\!64$ channels with adaptive average pooling) producing a
$128$-d feature, concatenated with the normalized agent position
$\mathrm{pos}/N\in[0,1]^3$, followed by a $(256,256)$ MLP head over $|\mathcal{A}|=64$
discrete actions. We train DQN ($\gamma{=}0.99$, replay buffer $5{,}000$, batch $64$,
target update every $10$ episodes, $\epsilon$ from $1.0$ to $0.05$, decay $0.995$,
Adam lr $10^{-3}$) and PPO ($\gamma{=}0.99$, $\lambda{=}0.95$, clip $0.2$, $4$ epochs,
batch $64$, rollout $512$, entropy $0.1$, Adam lr $10^{-4}$, grad clip $0.5$) for
$50{,}000$ episodes per condition with seed $42$.

\subsubsection{Rubik's Cube}
\label{app:exp:rubik}

\paragraph{Environment.}
We study the standard $N\!\times\! N\!\times\! N$ Rubik's cube for $N\in\{2,3,4,5,6\}$
with $12$ quarter-turn face moves $\{U,U',F,F',L,L',D,D',B,B',R,R'\}$, encoded as a
length-$6N^2$ sticker tensor with $6$ color labels and one-hot input dimension $36N^2$.
Training data are scrambles of length up to $30$ generated uniformly from the action
set; evaluation uses $50$ fresh scrambles per checkpoint with a fixed seed
 so each depth $d$ is tested on the
same scrambles every epoch.

\paragraph{Symmetry group.}
The cube admits the full rotational octahedral group $O_h\cong S_4$, of order
$|O_h|=24$, acting on the sticker tensor by permutations that simultaneously remap
faces and rotate within faces. We generate this group as
$O_h=\langle x,y\rangle$ with $x=R\!\circ\! L'$ and $y=U\!\circ\! D'$ — two perpendicular
$90^\circ$ whole-cube rotations — and BFS all $24$ elements at $N{=}2$
(\texttt{environment/cubeN.py}, \texttt{\_precompute\_rotations}). For general
$N$, each rotation is lifted by extracting (face-map, within-face inverse-rotation)
from its $N{=}2$ permutation and rebuilding the action on the $N\!\times\! N$ grid,
yielding a precomputed tensor \texttt{rotation\_perms} of shape $(24,6N^2)$.

\paragraph{Canonicalization.}
As illustrated in Fig.~\ref{fig:rubik-canon-schematic}, on every batch the environment computes
\[
r(s)\;=\;\arg\min_{g\in O_h}\;\mathrm{lex}\bigl(g\!\cdot\! s\bigr),
\]
i.e.\ the lexicographically smallest sticker string among the $24$ rotated copies
of $s$. This is implemented as a batched tensor op:
\texttt{rotated = states[:, rotation\_perms]} produces all $24$ rotated views
simultaneously, after which a chunked base-$6$ hash with progressive tie-breaking
($13$ stickers per chunk, $6^{13}\!\approx\!1.3\!\times\!10^{10}$ fits in
\texttt{int64}) selects the canonical index per state without ever forming the full
length-$6N^2$ key. The same canonicalization is applied to (i) the scrambled inputs,
(ii) all explored next states inside the DAVI target update, and (iii) the test-time
scrambles before A$^\star$ search. Because $r(g\!\cdot\! s)=r(s)$ for
all $g\in O_h$, the value head $V_\theta\circ r$ is automatically $O_h$-invariant
without any change to the network architecture.

\begin{figure}[t!]
    \centering
    \includegraphics[width=0.98\linewidth]{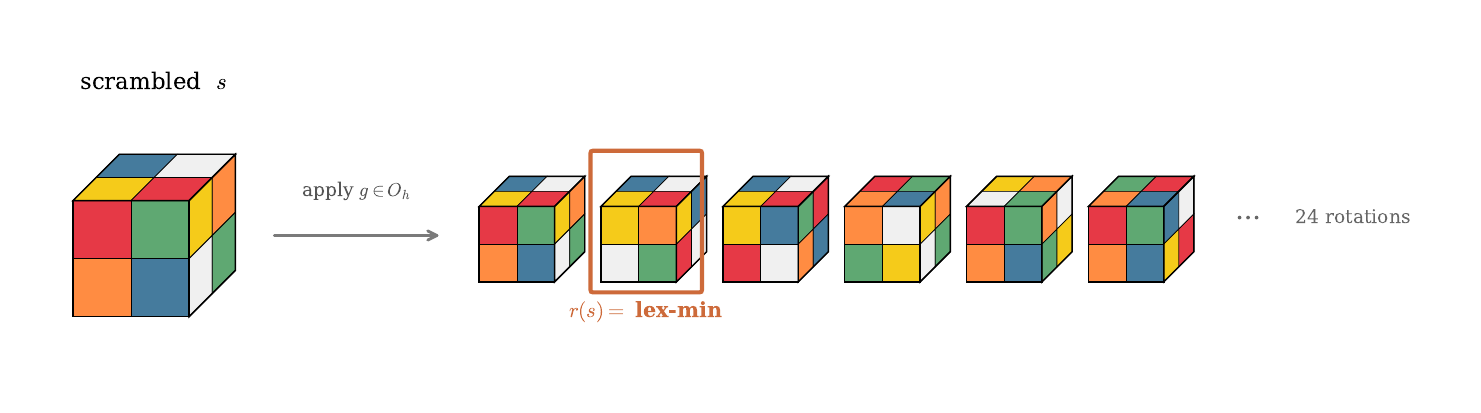}
    \caption{Rubik's cube canonicalization. A scrambled state $s$ (left) is replaced
    by the lexicographically smallest representative of its orbit
    $\mathcal{O}(s)=\{g\!\cdot\! s : g\in O_h\}$ under the $24$-element rotation group
    $O_h\cong S_4$}
    \label{fig:rubik-canon-schematic}
\end{figure}

\paragraph{Architecture and training.}
Both baseline and canonicalized variants share the network and hyperparameters from
\citep{agostinelli2019solving}.

\subsubsection{Mimicgen PPO with Flow-Matching Policy}
\label{app:exp:mimicgen}

\paragraph{Environment.}
We RL-finetune a pretrained flow-matching manipulation policy on Mimicgen
\citep{mandlekar2023mimicgen} tasks, namely \texttt{nut\_assembly\_d0}
and \texttt{stack\_three\_d1}. The action space is $10$-dimensional (xyz position $+$ rot6d
orientation $+$ gripper) with absolute pose control. Observations are a
wrist RGB image at $84{\times}84$, a $1024$-point colored point cloud,
end-effector pose (\texttt{eef\_pos}, \texttt{eef\_quat}), and gripper
qpos. The policy operates over a chunk of $H=16$ timesteps with
$n_{\mathrm{obs}}=2$ observation steps and $n_{\mathrm{act}}=8$ executed
action steps per call. We run $16$ parallel environments and evaluate on
$50$ held-out scenes per checkpoint.

\paragraph{Policy.}
The policy is the ManiFlow point-cloud architecture
\citep{yan2025maniflow}: a
\texttt{DP3Encoder} (PointNet-style, $128$-d feature) feeds a DiTX
transformer trunk ($12$ layers, $8$ heads, embed dim $768$, visual
condition length $1024$, with QK-norm) that parameterizes a $4$-step
rectified flow. Inference uses an SDE sampler with exploration noise
$\sigma=0.3$, so each denoising step produces a Gaussian transition
$x_{k+1}\!\sim\!\mathcal{N}(\mu_k(x_k,t_k),\,\sigma_k^2 I)$, and the
log-probability of an action chain is
\[
\log\pi(a\mid s)\;=\;\sum_{k=0}^{3}\log\mathcal{N}\bigl(x_{k+1};\,\mu_k,\sigma_k^2 I\bigr).
\]
A separate value head (AttnPool over the visual condition followed by an
MLP) shares the encoder. We initialize all weights from an
imitation-learning checkpoint and only update the DiTX trunk and value
head during RL.

\paragraph{Canonicalization.}
We retain the SE(3) canonicalization used during pretraining: the
environment expresses observations and actions in a canonical frame
defined by the robot pose.

\paragraph{Reward.}
We use the environment's binary task-success signal
$r^{\mathrm{env}}_t \in \{0,1\}$ scaled by a fixed bonus $\alpha = 10$:
\[
r_t \;=\; \alpha\, r^{\mathrm{env}}_t.
\]

\paragraph{PPO.}
We use a Gaussian PPO over the SDE chain
(\texttt{loss\_mode = "gaussian"}), in line with recent online RL fine-tuning
of flow-based policies \citep{chen2025pirl}. The importance ratio is
recomputed by replaying the stored chains $(x_0,\dots,x_4)$ through the
current DiTX:
\[
\rho \;=\; \exp\!\Bigl(\textstyle\sum_{k=0}^{3}\bigl[
  \log\mathcal{N}(x_{k+1};\mu_k^{\mathrm{new}},\sigma_k^2 I)
  -\log\mathcal{N}(x_{k+1};\mu_k^{\mathrm{old}},\sigma_k^2 I)\bigr]\Bigr),
\]
and the clipped surrogate uses $\epsilon=0.2$
\citep{schulman2017proximal}. GAE uses $\gamma=0.99$, $\lambda=0.95$, with
advantage normalization. The value loss is clipped to the same $0.2$
range with coefficient $0.5$. The entropy coefficient is $0$; we
early-stop the update when KL exceeds $0.03$. Each rollout collects $500$
episodes, followed by a single PPO epoch over that batch. We run $4$
critic-warmup rollouts ($3$ epochs each, value head only) before
unfreezing the actor.

\paragraph{Optimization.}
AdamW with actor lr $10^{-5}$ and critic lr $2\!\times\!10^{-4}$, linear
schedule with $10{,}000$ warmup steps, gradient-norm clipping at $0.5$,
minibatch size $16$, and gradient accumulation $128$ (effective batch
$2048$). The total budget is $2\!\times\!10^{6}$ environment steps;
evaluation runs every $3$ rollouts and we checkpoint every $100$
rollouts. Seed $42$.

\begin{figure}[t!]
    \centering

    \begin{subfigure}[b]{0.49\textwidth}
        \centering
        \hspace{5mm}
        \includegraphics[width=0.9\linewidth]{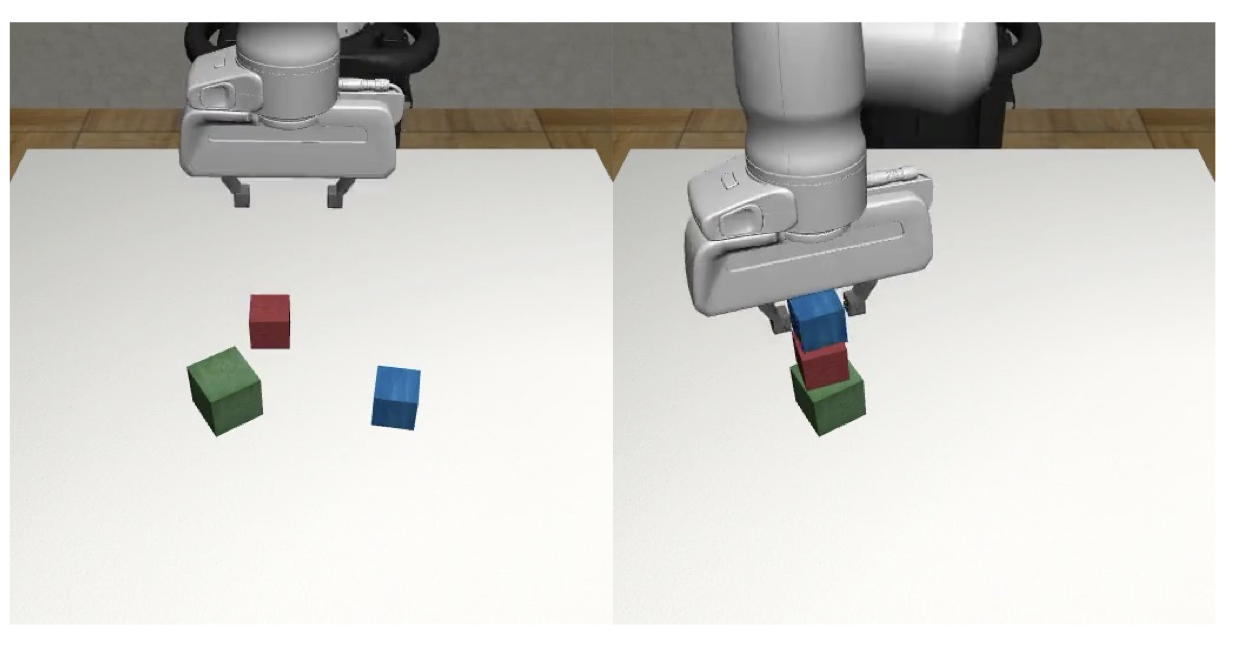}
    \end{subfigure}
    \hfill
    \begin{subfigure}[b]{0.49\textwidth}
        \centering
        \hspace{5mm}
        \includegraphics[width=0.9\linewidth]{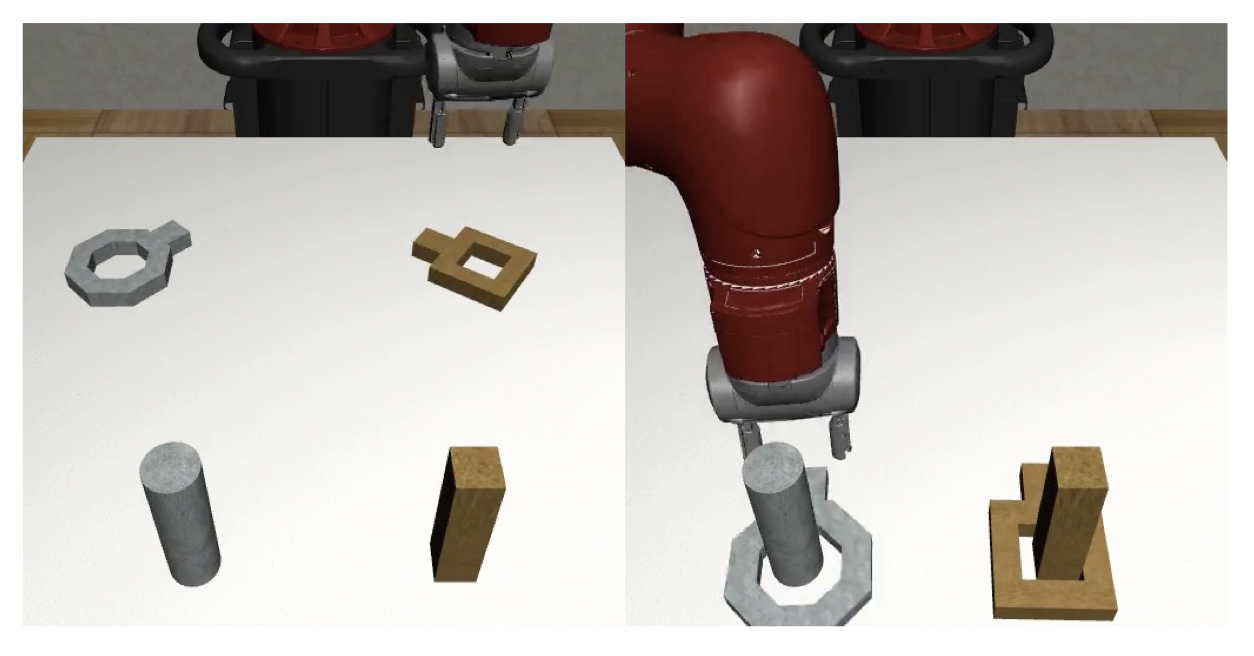}
    \end{subfigure}

    \caption{Illustrations of the block stacking and nut assembly robotic manipulation benchmarks.}
    \label{fig:rl-task-illustrations}
\end{figure}



\end{document}